\documentclass[10pt,journal,cspaper,compsoc]{IEEEtran}
\usepackage{url} 
\usepackage{amssymb,amsmath,amsfonts}
\usepackage{amsmath}
\usepackage{graphicx}
\usepackage{times}
 \usepackage{ifpdf}
 \usepackage{amsthm}
\usepackage{algorithm}
\usepackage[algo2e, ruled, vlined]{algorithm2e}

\newcommand{\ignore}[1]{}
\newcommand{\notinproc}[1]{#1}
\newcommand{\onlyinproc}[1]{}

\usepackage[pdftex]{xcolor}

\newtheorem{thm}{Theorem}[section]
\newtheorem{theorem}{Theorem}[section]

\newtheorem{lemma}[thm]{Lemma}

\newtheorem{example}[thm]{Example}

\newtheorem{corollary}[thm]{ Corollary}

\def\calA{\mathcal{A}}
\def\calX{\mathcal{X}}

\def\calD{\mathcal{D}}
\def\calQ{\mathcal{Q}}
\def\E{{\textsf E}}
\def\var{{\textsf var}}
\def\CV{{\textrm CV}}

\def\vecp{\boldsymbol{p}}
\def\vecm{\boldsymbol{m}}
\def\veca{\boldsymbol{a}}

\newcommand{\Sum}{\mathop{\textsf{sum}}}
\newcommand{\threshold}{\mathop{\textsf{thresh}}}
\def\thresh{\threshold}

\newcommand{\Capp}{\mathop{\textsf{cap}}}
\def\Cap{\Capp}

\def\Erlang{\textsf{Erlang}}

\def\eachg{\textsf{ForEach}}
\def\allg{\textsf{ForAll}}
\def\optg{\textsf{ForOpt}}
\title{Multi-Objective Weighted Sampling}

\author{Edith Cohen \\ 
Google Research\\
Mountain View, CA, USA \\
{\tt edith@cohenwang.com} 
}
\ignore{
\numberofauthors{1}
\author{
\alignauthor Edith Cohen\\
       \affaddr{Tel Aviv University, Israel \\ Google Research}
       \email{edith@cohenwang.com}
}
}

\begin{document}

\SetKwFunction{Hash}{Hash}
\SetKwFunction{keybase}{KeyBase}
\SetKwFunction{key}{Key}
\SetKwFunction{seed}{seed}

\IEEEcompsoctitleabstractindextext{
\begin{abstract}
\footnote{This is 
  a full version of a HotWeb 2015 paper}
         {\em Multi-objective samples} are powerful and versatile summaries of large data sets. 
For a set of keys $x\in X$ and associated values 
$f_x \geq 0$, a weighted sample taken with respect to 
$f$ allows us to approximate {\em segment-sum statistics}  $\Sum(f;H) =
   \sum_{x\in H} f_x$,  for any subset $H$ of the keys, with statistically-guaranteed quality that depends on sample size and the relative weight of $H$.
When estimating $\Sum(g;H)$ for $g\not=f$, however, quality guarantees
are lost.
A multi-objective sample with respect to a set of functions $F$ 
provides for each $f\in F$ the same statistical guarantees
as a dedicated weighted sample while minimizing the summary size.

We analyze properties of multi-objective samples and present sampling schemes and meta-algortithms
for estimation and optimization while showcasing two important
application domains.
 The first are key-value data sets, where different functions $f\in F$ applied to the values
 correspond to different statistics such as moments,
thresholds, capping, and sum. A multi-objective
sample allows us to approximate all statistics in $F$. 
The second is metric spaces, where keys are points, and each $f\in F$
is defined by a set of points $C$ with $f_x$ being the service
cost of $x$ by $C$, and $\Sum(f;X)$ models
centrality or clustering cost of $C$.
 A multi-objective sample  allows us to estimate costs for each $f\in
 F$. In these domains, multi-objective samples are often of small size, 
are efficiently to construct, and enable scalable estimation and optimization.
We aim here to facilitate further
applications of this powerful technique.

\end{abstract}
}
\maketitle
\section{Introduction}

 Random sampling is a powerful tool for working with very large data sets on which exact computation, even of simple statistics, can be time and
resource consuming.  A small sample of the
data allows us to efficiently obtain approximate answers.

Consider data in the form of key value pairs 
 $\{(x,f_x)\}$, where keys $x$ are from some universe $\calX$,
 $f_x \geq 0$, and we define $f_x \equiv 0$ for keys $x\in \calX$ that
are not present in the data.
Very common statistics over such data are 
{\em segment sum statistics}
$$\Sum(f;H)=\sum_{x\in H} f_x\ ,$$ where 
$H\subset \calX$ is a {\em segment}\notinproc{\footnote{
The alternative term {\em selection} is used in the DB literature and 
the  term {\em domain} is used in the 
statistics literature. 
}}
 of $\calX$.
Examples of such data sets are IP flow keys and bytes, 
users and activity, 
or customers and distance to the nearest facility.
Segments may correspond to a certain demographic or location or 
other meta data of keys. Segment 
statistics in these example correspond respectively 
to total traffic, activity, or service cost of the segment.

When the data set is large, we can compute a weighted sample which
includes each key $x$ with probability (roughly) proportional to $f_x$ and allows us to estimate segment sum statistics  $\widehat{\Sum}(f;H)$ for query segments $H$.
Popular weighted sampling schemes \cite{SSW92,Tille:book}  include
Poisson Probability Proportional 
  to Size (pps) \cite{HansenH:1943},
VarOpt \cite{Cha82,varopt_full:CDKLT10}, and the bottom-$k$  
schemes \cite{Rosen1997a,bottomk07:ds,bottomk:VLDB2008}
Sequential Poisson  (priority) \cite{Ohlsson_SPS:1998,DLT:jacm07} and 
PPS without replacement (ppswor) \cite{Rosen1972:successive}.

 These weighted samples provide us with nonnegative unbiased estimates
 $\widehat{\Sum}(g;H)$ for any segment and $g\geq 0$ (provided that
$f_x >0$ when $g_x >0$).
For all statistics $\Sum(f;H)$ we obtain
statistical guarantees on estimation quality: The error, measured by the
coefficient of variation (CV), which is the standard deviation divided by the mean, is
at most the inverse of the square root of the size of the sample
multiplied by the fraction $\Sum(f;H)/\Sum(f,{\cal X})$ of `` weight'' that is due to the segment $H$.
This trade-off of quality (across segments) and sample size
are (worst-case) optimal. Moreover, the estimates are
well-concentrated in the Chernoff-Bernstein sense:  The
probability of an error that is $c$ times the CV  decreases
exponentially in $c$.

 In many applications, such as the following examples,  there are
 {\em multiple} sets of values $f\in F$ that are associated with the
 keys:
(i) Data records can come with explicit multiple weights, as with
activity summaries of customers/jobs that specify both bandwidth
and computation consumption.
(ii) Metric objectives, such as our service cost example, where each 
configuration of facility locations induces a different set of distances and
hence service costs. 
(iii) The raw data can be specified in terms of a set of key value
pairs $\{(x,w_x)\}$ but we are interested in different functions $f_x
\equiv f(w_x)$ of the values that correspond to different statistics
such as

\smallskip
\begin{tabular}{l|l}
\textcolor{blue}{Statistics} & \textcolor{blue}{function $f(w)$} \\
\hline
 {\em count} & $f(w) =1$ for $w>0$ \\
 {\em sum} & $f(w)=w$ \\
 {\em threshold} with
$T>0$ & $\threshold_T(w)=I_{w \geq T}$ \\
{\em moment} with
$p>0$ & $f(w)=w^p$ \\
{\em capping} with  $T>0$ & $\Cap_T =\min\{T,w\}$
\end{tabular}

\begin{example}
Consider a toy data set \textcolor{blue}{$\calD$}:\\
\textcolor{blue}{
$(u1,5)$, $(u3,100)$, $(u10,23)$, $(u12,7)$, $(u17,1)$,
$(u24,5)$, $(u31,220)$, $(u42,19)$, $(u43,3)$, $(u55,2)$}\\
For a segment $H$ with
$H\cap {\cal D}=\{$\textcolor{blue}{$u3$},
\textcolor{blue}{$u12$}, \textcolor{blue}{$u42$},
\textcolor{blue}{$u55$}$\}$, we have
\textcolor{purple}{$\text{sum}(H) = 128$},
\textcolor{purple}{$\text{count}(H) = 4$},
\textcolor{purple}{$\threshold_{10}(H) = 2$},
\textcolor{purple}{$\Cap_{5}(H) = 17$}, and
\textcolor{purple}{$\text{2-moment}(H) = 10414$}.
\end{example}

For these applications, we are interested in a summary
that can provide us with
estimates with statistically guaranteed quality for each $f\in F$.
The naive solutions are not satisfactory: 
We can compute a weighted
sample taken with respect to a particular  $f\in F$, but the quality of the estimates $\Sum(g;H)$ rapidly degrades
with the dissimilarity between $g$ and $f$.
We can compute a dedicated sample for each 
$f\in F$, but the total summary size can be much larger than necessary.
 {\em Multi-objective} samples, a notion crystallized
in \cite{multiw:VLDB2009}\footnote{The 
collocated model},
provide us with the desired statistical guarantees on quality
with minimal summary size.

Multi-objective samples build on the classic notion of 
{\em sample coordination}
\cite{KishScott1971,BrEaJo:1972,Saavedra:1995,ECohen6f, Rosen1997a,Ohlsson:2000}.  In a nutshell, coordinated samples are locality sensitive hashes of
$f$, mapping similar $f$ to similar samples.
A multi-objective sample is (roughly) the union
$S^{(F)} = \bigcup_{f \in F}  S^{(f)}$  
of all the keys that are included
in coordinated weighted samples $S^{(f)}$ for $f\in F$.
Because the samples are coordinated, the number of distinct keys included and hence the size of $S^{(F)}$ is (roughly) as small as possible.
Since for each $f\in F$, the sample $S^{(F)}$ ``includes'' the dedicated sample
$S^{(f)}$, the estimate quality from $S^{(F)}$ dominates that of $S^{(f)}$.

\ignore{
Multi-objective samples can be considered
in any scenario where there are multiple sets $f\in F$ of value
$f_x \geq 0$ assigned to a set $x\in X$ of keys 
(i.e, the values need not correspond to function of the form $f(w_x)$).
Another important
application domain are metric objectives, such as the cost of
clustering or facility location.
Here we have a set $X$ of point in a metric space.
Each {\em configuration}  $f$ of cluster centers or facility placement induces a
cost $f_x \geq 0$ for each point $x\in X$.   We are interested in estimating the total cost $V(f) = \sum_x f_x$ of the configuration $f$.  To do so for a particular configuration,  we can use a small weighed sample of points $S\subset X$ according (roughly) to $f_x$ (leaving aside for now the question on how to obtain such a sample without knowing the full set $f$).
The sample $S$ is a summary of $X$ that
can be used to estimate the cost of other configurations $g\not= f$, but the quality degrades again with the dissimilarity between $g$ and $f$.
When we are interested in a set $F$ of different configurations, and in
a summary $S$ that would allow us to
estimate $V(f)$ for each $f\in F$ with quality guarantees, we can use a multi-objective sample.
}

 In this paper, we review the definition of multi-objective samples,
 study their properties, and present efficient sampling schemes.  We
 consider both general sets $F$ of objectives and  families $F$
 with special structure.  By exploiting special structure, we can bound
 the {\em overhead}, which is the increase factor in sample size
 necessary to meet multi-objective quality guarantees, and obtain
efficient sampling schemes that avoid dependence of the computation on $|F|$.

\subsection{Organization}
In Section~\ref{single:sec}
we review (single-objective) weighted sampling 
focusing on the Poisson pps and bottom-$k$ sampling schemes.
We then review the definitions  \cite{multiw:VLDB2009}
and establish properties of multi-objective samples.
In Section~\ref{multi:sec} we study multi-objective pps samples.
We show that the multi-objective sample size is also
{\em necessary} for meeting  the quality guarantees for segment
statistics for all $f\in F$.  We also show that the guarantees are met
when we use upper bounds on the multi-objective pps sampling
probabilities instead of working with the exact values.
In Section~\ref{bottomkMO:sec} we study 
multi-objective bottom-$k$ samples.

In Section~\ref{basis:sec} we establish a fundamental property of
multi-objective samples:  We define the {\em sampling closure}
$\overline{F}$ of a set of objectives $F$, as 
all functions $f$ for which a multi-objective sample $S^{(F)}$ meets
the quality guarantees for segment statistics.
Clearly $F \subset \overline{F}$ but we 
show that the closure $\overline{F}$ also includes
every  $f$ that is a non-negative linear combination of functions from $F$.

In Section \ref{monotone:sec}, we consider data sets in the form of
key value pairs and the 
family $M$ of all monotone non-decreasing functions of the values.
This family includes most natural statistics, such as our examples of count, sum, threshold,
moments, and capping.
Since $M$ is infinite, it is inefficient to apply a generic multi-objective
  sampling algorithm to compute $S^{(M)}$.
We present efficient near-linear sampling
  schemes for $S^{(M)}$ which also apply over streamed or
distributed data.    Moreover, we establish a  bound on the
sample size of  $\E[|S^{(M)}|] \leq
k \ln n$,  where  $n$ is the number of keys in our data set and $k$ is
the reference size of the single-objective samples $S^{(f)}$ for each
$f\in M$.  The design is based on a surprising relation  to 
All-Distances Sketches \cite{ECohen6f,ECohenADS:TKDE2015}.
Furthermore, we establish that (when key weights are unique), a sample of size 
$\Omega(k \ln n)$ is necessary: Intuitively, the ``hardness''  stems from the need to 
support all threshold functions.

In Section \ref{allcaps:sec} we study  the set
 $C = \{\Cap_T \mid T>0\}$
of all capping functions.  The closure $\overline{C}$ includes
all concave $f\in M$ with at most a linear growth (satisfy $f'(x) \leq 1$ and
$f''(x) \leq 0$).
Since $C \subset M$,  the multi-objective sample $S^{(M)}$ includes
$S^{(C)}$ and provides estimates with statistical
guarantees for all $f\in \overline{C}$.  The more specialized sample
$S^{(C)}$, however,  can be much smaller than $S^{(M)}$.
We design an efficient algorithm for computing  $S^{(C)}$ samples.

In Section~\ref{metric:sec} we discuss metric objectives and 
multi-objective samples as summaries of a set of points that allows us
to approximate such objectives.
\ignore{
Our keys here are a set of points $X$ in a metrics space $M$. 
Our ``functions'' $f_Y:X$ are indexed by a subset of points
$Y \subset M$ so that  $f_Y(x)$ is a function of the 
set of distances $\{d(y,x) \mid y\in Y\}$.  For example,  $f_Y(x) =
d(Y,x)^2$ is the squared distance of $x$ from $Y$.  
We are interested in a summary from which we can 
estimate $Q(Y) = \sum_{x\in X} f_Y(x))$  when $Y$ is taken from a
set $\cal Y$.
}

In Section~\ref{eachall:sec}  we discuss different types of 
statistical guarantees across functions $f$ in settings where we are
only interested in statistics $\Sum(f ; \calX)$ over the full data.  
 Our basic multi-objective samples analyzes the sample size required
 for \eachg, where the
 statistical guarantees apply to each estimate $\widehat{\Sum}(f;H)$
 in isolation.  In particular, also to each estimate over
 the full data set.
\allg\ is much stronger and bounds 
the (distribution of) the maximum relative error of 
estimates $\widehat{\Sum}(f ; \calX)$ for all $f\in F$.  Meeting \allg\ typically 
necessitates a larger multi-objective sample size than meeting \eachg.

In section~\ref{opt:sec} we present a meta-algorithm for optimization
over samples.  The goal is to maximize a (smooth) function of
$\Sum(f ; \calX)$ over $f\in F$.  When $\calX$ is large, we can
instead perform the optimization over a small multi-objective sample
of $\calX$.   This framework has important
applications to metric objectives and estimating loss of a model from examples.
 The \optg\ guarantee is for a sample size that facilitates such
 optimization, that is, the approximate maximizer over the sample is an
 approximate maximizer over the data set.  This guarantee is stronger
 than \eachg\ but generally weaker than \allg.  We make a key
 observation that with a
\eachg\ sample we are only prone to testable one-sided errors on the
optimization result.  Based on that, we present an
 adaptive algorithm where the sample size is increased until
 \optg\ is met.  This framework unifies and generalizes previous work of
 optimization over coordinated samples \cite{topk:conext06,binaryinfluence:CIKM2014}.

We conclude in Section~\ref{conclu:sec}.

\onlyinproc{Due to the page limitation, many details are deferred to
  the full version of the paper (\url{http://arxiv.org/abs/1509.07445}).}


\section{Weighted sampling (single objective)} \label{single:sec}

We review weighted sampling schemes with respect to 
a set of values $f_x$, focusing on preparation for the
multi-objective generalization.     The schemes are specified in terms of
a sample-size parameter $k$ which allows us to trade-off 
representation size and estimation quality.

\subsection{Poisson  Probability Proportional to Size (pps)}
The pps sample $S^{(f,k)}$ includes
each key $x$ independently with probability
\begin{equation} \label{ppsprob:eq}
p^{(f,k)}_x = \min\{1, k \frac{f_x}{\sum_y f_y}\} \ .
\end{equation}

 \begin{example} \label{example2}
The table below lists pps sampling probabilities $p^{(f,3)}_x$ ($k=3$,
rounded to the
nearest hundredth)  for keys in our example data for
sum ($f_x = w_x$), $\threshold_{10}$ ($f_x = I_{w_x \geq 10}$), and $\Cap_5$ ($f_x = \min\{5,w_x\}$).
The number in parenthesis is 
$\Sum(f,\calX)=\sum_x f_x$. We can see
that sampling probabilities highly vary between functions $f$.\\\vspace{0.03in}\\
{\tiny
\begin{tabular}{l||r|r|r|r|r|r|r|r|r|r}
\hline
 key &   \textcolor{blue}{u1} & \textcolor{blue}{u3} & \textcolor{blue}{u10} & \textcolor{blue}{u12} &\textcolor{blue}{u17} & \textcolor{blue}{u24} & \textcolor{blue}{u31} & \textcolor{blue}{u42} & \textcolor{blue}{u43}  &\textcolor{blue}{u55} \\
 $w_x$ & \textcolor{blue}{$5$} & \textcolor{blue}{$100$} & \textcolor{blue}{$23$} & \textcolor{blue}{$7$} & \textcolor{blue}{$1$} & \textcolor{blue}{$5$} & \textcolor{blue}{$220$} & \textcolor{blue}{$19$} & \textcolor{blue}{$3$} & \textcolor{blue}{$2$}
  \\
\hline
\textcolor{purple}{$\Sum$} (385)   & 0.04  & 0.78  & 0.18 & 0.05 &
0.01 & 0.04 &  1.00 &
0.15 & 0.02 & 0.02\\
\textcolor{purple}{$\threshold_{10}$} (4) & 0.00  & 0.75  & 0.75 &
0.00  & 0.00 & 0.00 & 0.75 & 0.75 & 0.00& 0.00   \\
 \textcolor{purple}{$\Cap_5$}  (41)  & 0.37 & 0.37  & 0.37 & 0.37 &
 0.07 & 0.37 & 0.37 & 0.37  & 0.22  & 0.15 
\end{tabular}
 }
 \end{example}

PPS samples can be computed by association a random
 value $u_x \sim U[0,1]$ with each key $x$ and including the key in
 the sample if $u_x \leq p_x^{(f,k)}$.
 This formulation to us when there are multiple objectives as it facilitates the coordination of samples taken with respect to the different objectives.
Coordination is achieved using the same set $u_x$.

\subsection{Bottom-$k$ (order) sampling} Bottom-$k$ sampling unifies
priority (sequential Poisson)
  \cite{Ohlsson_SPS:1998,DLT:jacm07} and pps without replacement (ppswor) sampling \cite{Rosen1972:successive}. 
To obtain a bottom-$k$ sample for $f$ 
we associate a random value $u_x \sim U [0,1]$ with 
each key.   To obtain a ppswor sample we use $r_x \equiv -\ln(1-u_x)$
and to obtain a priority sample we use  $r_x 
\equiv u_x$.  The bottom-$k$ sample $S^{(f,k)}$ for $f$ contains 
the $k$ keys with minimum {\em  $f$-seed}, where
$$f\text{-}\seed{x} \equiv \frac{r_x}{f_x}\ .$$
To support estimation, we also retain the {\em threshold},
$\tau^{(f,k)}$, which is defined to be the $(k+1)$st smallest $f$-seed.

\subsection{Estimators}
We estimate a statistics $\Sum(g;H)$ from a weighted sample $S^{(f,k)}$
using the {\em inverse probability} estimator \cite{HT52}:
\begin{equation} \label{ppsest}
\widehat{\Sum}(g;H) = \sum_{x\in H\cap S}  \frac{g_x}{p^{(f,k)}_{x}}\ . 
\end{equation}
The estimate is always nonnegative and is unbiased when the
functions satisfy $g_x >0 \implies f_x >0$ (which ensures that any key $x$
with $g_x>0$ is sampled with positive probability).
To apply this estimator, we need to compute $p^{(f,k)}_{x}$ for $x\in S$.  To do so with pps samples \eqref{ppsprob:eq}
we include the sum $\sum_x f_x$ with $S$  as auxiliary information.

For  bottom-$k$ samples, inclusion probabilities of keys are
not readily available.  We therefore use the inverse probability
estimator \eqref{ppsest} with
{\em conditional} probabilities $p^{(f,k)}_{x}$
  \cite{DLT:jacm07,bottomk:VLDB2008}:
A key $x$, fixing the randomization $u_y$ for all other keys, is
sampled if and only if $f\text{-}\seed{x} < t$, where $t$ is the $k$th 
smallest $f$-seed among keys $y\not= x$.   
For $x\in S^{(f,k)}$, the $k$th smallest $f$-seed among
other keys is $t=\tau^{(f,k)}$, and thus
\begin{equation}\label{inclusiont}
p^{(f,k)}_{x} = \Pr_{u_x \sim U[0,1]}\left[\frac{r_x}{f_x}<  \tau^{(f,k)}\right]\ .
\end{equation}
Note that the right hand side expression for probability  is equal to
$1-e^{-f_x t }$ with ppswor and to $\min\{1, f_x t \}$ with
priority sampling.

\ignore{
\begin{equation} \label{ppsworest}
\hat{Q}(g;H) = \sum_{x\in H\cap S}  \frac{g(w_x)}{p^{(f, \tau^{(f,k)})}_x}\ . 
\end{equation}
}

\subsection{Estimation quality}

We consider the variance and concentration of our estimates.
A natural measure of estimation quality of our unbiased estimates is
the coefficient of variation (CV),
which is the ratio of the standard deviation to the mean. 
We can upper bound the CV  of our estimates \eqref{ppsest}
of $\Sum(g;H)$ in terms of the (expected) 
sample size $k$ and the relative $g$-weight  of the segment $H$,
defined as
$$q^{(g)}(H) = \frac{\Sum(g;H)}{\Sum(g;\calX)}\ .$$
To be able to express a bound on the CV when we estimate 
a statistics $\Sum(g;H)$ using a weighted sample taken with respect to 
$f$, we define the {\em disparity}  between $f$ and $g$ as 
$$\rho(f,g) = \max_{x}\frac{f_x}{g_x}   \max_{x}\frac{g_x}{f_x}\ .$$
The disparity always satisfies $\rho(f,g) \geq 1$ and we have equality
$\rho(f,g)=1$ only when $g$ is a scaling of $f$, that is, equal to $g=c f$
for some $c>0$.  We obtain the following upper bound:



\begin{theorem} \label{cvbound:thm}
For pps samples and the estimator \eqref{ppsest},
$$\forall g \forall H,\ \CV[\widehat{\Sum}(g;H)] \leq \sqrt{\frac{\rho(f,g)}{q^{(g)}(H) k }}\ .$$  
For bottom-$k$ samples, we replace $k$ by $k-1$.
\end{theorem}
The proof for $\rho=1$ is standard for pps,
provided in \cite{ECohen6f,ECohenADS:TKDE2015} for ppswor, and in
\cite{Szegedy:stoc06} for priority samples.
The proof for $\rho\geq 1$ for ppswor is provided in
Theorem \ref{ppsworcvrho:thm}.
 The proof for pps is simpler, using a subset of the arguments.  The proof for
priority can be obtained by generalizing \cite{Szegedy:stoc06}.

 Moreover, the estimates obtained from these 
weighted sample are  {\em concentrated} in the 
Chernoff-Hoeffding-Bernstein sense.  We provide the proof for the multiplicative form of the bound and for Poisson pps samples:
\begin{theorem} \label{concentration:lemma}
  For $\delta \leq 1$,
  \begin{eqnarray*}
    \lefteqn{\Pr[|\widehat{\Sum}(g;H)-\Sum(g;H)| > \delta \Sum(g;H)]} \\
    &\leq& 2\exp(-q^{(g)}(H) k\rho^{-2} \delta^{2}/3)\ .
   \end{eqnarray*}
  For $\delta > 1$,
  \begin{eqnarray*}
    \lefteqn{\Pr[\widehat{\Sum}(g;H)-\Sum(g;H) > \delta \Sum(g;H)]}\\
    &\leq& \exp(-q^{(g)}(H) k\rho^{-2} \delta/3)\ .
  \end{eqnarray*}
\end{theorem}
\begin{proof}
  Consider Poisson pps sampling and the inverse probability estimator.
  The contribution of keys that are sampled with $p^{(f,k)}_x  =1$ is computed exactly.
  Let the contribution of these keys be $(1-\alpha)\Sum(g;H)$, for some $\alpha\in [0,1]$).  If $\alpha=0$, the estimate is the exact sum and we are done.
  Otherwise, it suffices to estimate the remaining $\alpha \Sum(g;H)$ with relative error   $\delta' = \delta/\alpha$. 

  Consider the remaining keys, which 
  have inclusion probabilities $p^{(f,k)}_x \geq k f_x/\Sum(f)$.
  The contribution of such a key $x$
  to the estimate is $0$ if $x$ is not sampled and is
  $$\frac{g_x}{p^{(f,k)}_x} \leq \rho(f,g)  \Sum(f)/k $$ when $x$ is sampled.
  Note that by definition $$\Sum(g;H) = q^{(g)}(H) \Sum(g) \geq q^{(g)}(H) \Sum(f)/\rho(f,g)\ .$$
  We apply the concentration bounds to the sum of random variables in the range $[0,  \rho(f,g)  \Sum(f)/k]$.  To use the standard form, we can normalize our random variables and to have range $[0,1]$ and accordingly normalize the expectation $\alpha \Sum(g;H)$ to obtain
  $$\mu = \alpha \frac{\Sum(g;H)}{\rho(f,g)  \Sum(f)/k]} \geq \alpha q^{(g)}(H) \rho(f,g)^{-2} k\ .$$
    We can now apply multiplicative Chernoff bounds for random variables in the range $[0,1]$ with $\delta' = \delta/\alpha$ and expectation $\mu$.
    The formula bounds the probability of relative error that exceeds $\delta$ by $2\exp(-\delta^2\mu/3$ when $\delta<1$ and by $\exp(-\delta\mu/3$ when $\delta>1$.
\end{proof}

\subsection{Computing the sample}
Consider data presented as streamed or distributed {\em elements} of the
form of key-value pairs $(x,f_x)$, where
$x\in \calX$ and  $f_x >0$.   
We define $f_x \equiv 0$ for keys $x$ that are not in the data.

An important property of our samples (bottom-$k$ or pps) is
that they are {\em composable} (mergeable). Meaning that a sample of the
union of two data sets can be computed from the samples of the data sets.
Composability facilitates efficient streamed or distributed  computation.
The sampling algorithms can use a 
random hash function applied to key $x$  to generate $u_x$ -- so 
seed values can be computed on the fly from $(x,f_x)$ and do not need
to be stored.

  With bottom-$k$ sampling we permit 
keys $x$ to occur in multiple elements, in which case we
define $f_x$ to be the maximum value of elements with key
$x$.  
The sample $S(\calD)$ of a set $\calD$ of elements contains the pair $(x,f_x)$
for the $k+1$ (unique) keys with smallest $f$-seeds \footnote{When keys are unique to
  elements it suffices to keep only the $(k+1)$st smallest $f$-seed
  without the pair $(x,f_x)$.}
 The sample of the union $\bigcup_i \calD_i$ 
is obtained from $\bigcup_i S(\calD_i)$ by first replacing multiple
occurrences of a key with the one with largest $f(w)$ and then 
returning the pairs for the $k+1$
keys with smallest $f$-seeds.

 With pps sampling, 
the information we store with our sample $S(\calD)$ includes the 
sum $\Sum(f;\calD) \equiv \sum_{x\in D}  f_x$ and the 
sampled pairs $(x,f_x)$, which are those with 
$u_x \leq k f_x/\Sum(f;\calD)$.   Because we need to accurately track the sum, we 
require that elements  have unique  keys. 
The sample of a union $D= \bigcup_i \calD_i$
is obtained using the sum $\Sum(f;\calD) = \sum_i \Sum(f;\calD_i)$, and retaining only keys in
$\bigcup_i S(\calD_i)$ that satisfy $u_x \leq k f_x/\Sum(f;\calD)$.

\section{Multi-objective pps samples} \label{multi:sec}

Our objectives are specified as pairs $(f,k_f)$ where $f \in F$ is a function and $k_f$ specifies a desired estimation quality
for $\Sum(f;H)$ statistics, stated in terms of 
the quality (Theorem~\ref{cvbound:thm} and
Theorem~\ref{concentration:lemma})
provided by a single-objective sample for $f$ with size parameter
$k_f$.  To simplify notation, we sometimes omit $k_f$ when clear from
context. 

A multi-objective sample $S^{(F)}$ \cite{multiw:VLDB2009}  is  defined by 
considering dedicated samples $S^{(f,k_f)}$ for each objective that
are {\em coordinated}.  The dedicated samples are coordinating by 
using the same randomization, which is the association of $u_x \sim U[0,1]$ with keys.
 The multi-objective sample $S^{(F)}=\bigcup_{f\in  F} S^{(f,k_f)}$ 
contains all keys that are included in at least one of the coordinated
dedicated samples.  In the remaining part of this section we study pps samples.  Multi-objective bottom-$k$ samples are studied in the next section.

\begin{lemma}
A multi-objective pps sample for $F$ includes each key 
$x$ independently  with probability
\begin{equation}\label{PPPSMOprob}
p^{(F)}_x = \min\{1, \max_{f\in F} \frac{k_f f_x}{\sum_y 
  f_y}\} \ .
\end{equation}
\end{lemma}
\begin{proof}
Consider coordinated dedicated pps samples for $f\in F$ obtained using the 
same set $\{u_x\}$.  The key $x$ is included in at least one of the
samples if and only if
the value $u_x$ is at most 
 the maximum  over objectives $(f,k_f)$ of the pps inclusion 
 probability for that objective:
\begin{eqnarray*}
u_x &\leq& \max_{f\in F} p_x^{(f,k_f)} = \max_{f\in F} \min\{1,\frac{k_f f_x}{\sum_y 
  f_y}\} \\
&=&  \min\{1, \max_{f\in F} \frac{k_f f_x}{\sum_y 
  f_y}\}  \ .
\end{eqnarray*}
Since $u_x$ are independent, so are the inclusion probabilities of
different keys.
\end{proof}

\begin{example}
Consider the three objectives: $\Sum$, $\threshold_{10}$, and $\Cap_5$
all with $k=3$ as in 
Example~\ref{example2}.   The expected size of $S^{(F)}$ is $|S^{(F)}| =
\sum_x p^{(F)}_x = 4.68$. The naive solution of maintaining a separate
dedicated sample for each objective would have total expected size  $8.29$  (Note that the
dedicated expected sample size for $\Sum$ is $2.29$ and for $\threshold_{10}$, and $\Cap_5$ it is $3$.
\end{example}

 To estimate a statistics $\Sum(g;H)$ from $S^{(F)}$,  we 
apply the inverse probability estimator 
\begin{equation} \label{MOest}
\widehat{\Sum}(g;H) = \sum_{x\in S^{(F)} \cap H} \frac{g_x}{p_x^{(F)}}\ . 
\end{equation}
using the probabilities 
$p^{(F)}_x$ \eqref{PPPSMOprob}.

To compute the estimator \eqref{MOest}, we need to know $p^{(F)}_x$ when  $x\in S^{(F)}$.
These probabilities can be computed if we
 maintain the sums $\Sum(f)=\sum_{x} f_x$ for $f\in F$ as auxiliary information and we have $f_x$ available to us when $x\in S^{(f,k_f)}$.

 In some settings it is easier to obtain upper bounds
 $\pi_x \geq p^{(F)}_x$ on the multi-objective pps inclusion probabilities, compute a Poisson sample using $\pi_x$, and apply the respective inverse-probability estimator 
\begin{equation} \label{UBest}
\widehat{\Sum}(g;H) = \sum_{x\in S \cap H} \frac{g_x}{\pi_x}\ . 
\end{equation}
 
\smallskip
\noindent
{\bf Side note:}
It is sometime useful to use other sampling schemes, in particular,
VarOpt (dependent) 
sampling~\cite{Cha82,GandhiKPS:jacm06,varopt_full:CDKLT10} to obtain a
fixed sample size. The estimation quality bounds on the CV and concentration also hold with VarOpt (which has negative covariances).

\subsection{Multi-objective pps estimation quality}
We show that the estimation quality, in terms of the bounds on the
CV and concentration, of
the estimator \eqref{MOest} is at least as good as that of the estimate we obtain from the dedicated samples. To do so we prove a more general claim that holds
for any Poisson sampling scheme that includes each key $x$ in the sample $S$
with probability
$\pi_x \geq p^{(f,k_f)}$ and the respective inverse probability estimator
\eqref{UBest}.

The following lemma shows that estimate quality can only improve when inclusion probabilities increase:
\begin{lemma} \label{domprob:lemma}
  The variance $\var[\widehat{\Sum}(g;H)]$ of \eqref{UBest} and hence $\CV[\widehat{\Sum}(g;H)]$ are non-increasing in $\pi_x$.
\end{lemma}
\begin{proof}
For each key $x$ consider the inverse probability estimator
$\hat{g}_x = g_x/\pi_x$ when $x$ is sampled and $\hat{g}_x=0$
otherwise.  Note that
$\var[\hat{g}_x] = g_x^2(1/\pi_x -1)$, which is decreasing with $\pi_x$.
We have
$\widehat{\Sum}(g;H)]=\sum_{x\in H} \hat{g}_x$.
When covariances between $\hat{g}_x$ are nonpositive, which is the
case in particular with independet inclusions, we have
$\var[\widehat{\Sum}(g;H)]= \sum_{x\in H} \var[\hat{g}_x]$ and the
claim follows as it applies to each summand.
\end{proof}

We next consider concentration of the estimates.
\begin{lemma}
  The concentration claim of Theorem~\ref{concentration:lemma} carries over when we use the inverse probability estimator with any sampling probabilities that satisfy $\pi_x \geq p^{(f,k)}_x$ for all $x$.
\end{lemma}
\begin{proof}
  The generalization of the proof is immediate, as the range of the random variables $\hat{f}_x$ can only decrease when we increase the inclusion probability.
\end{proof}

\subsection{Uniform guarantees across objectives}
 An important special case is when $k_f = k$ for all $f\in F$, that is, we
 seek uniform statistical guarantees for all our objectives.
We use the notation $S^{(F,k)}$ for the respective 
multi-objective sample. 
We can write the multi-objective pps probabilities \eqref{PPPSMOprob}
 as
\begin{equation} \label{uPPSMO}
 p_x^{(F,k)} = \min\{1,k \max_{f\in F} \frac{f_x}{\sum_y 
   f_y}\}= \min\{1,k p^{(F,1)}_x\} .
\end{equation}
The last equality follows when recalling the definitions of
$p^{(f,1)}_x = f_x/\sum_y f_y$ and hence
$$p^{(F,1)}_x = \max_{f\in F} p^{(f,1)}_x = \max_{f\in F} \frac{f_x}{\sum_y 
  f_y}\ . $$

We refer to $p_x^{(f,1)}$ as the {\em base pps} probabilities for $f$.
Note that the base pps probabilities are rescaling of $f$ and pps probabilities are invariant to this scaling.
We refer to $p_x^{(F,1)}$ as the
{\em multi-objective base pps} probabilities for $F$.
Finally, for a reason that will soon be clear, we refer to 
the sum $$h^{(\text{pps})}(F) \equiv \sum_x p^{(F,1)}_x $$ as
the multi-objective pps {\em overhead} of $F$.

It is easy to see that 
$h(F)\in [1,|F|]$  and $h(F)$ is closer to $1$ when all objectives in $F$ are more similar.
We can write
$$ k p^{(F,1)}_x = (k h(F)) \frac{p^{(F,1)}_x}{\sum_x p^{(F,1)}_x}\ .$$
That is, the multi-objective pps probabilities \eqref{uPPSMO} are equivalent to
single-objective pps probabilities with size parameter
$k h(F)$ computed with respect to base probability ``weights''
$g_x = p^{(F,1)}_x$.

\smallskip
\noindent
{\bf Side note: }
We can apply any single-objective weighted sampling
scheme such as VarOpt or bottom-$k$ to the weights $p^{(F,1)}_x$, or
to upper-bounds $\pi_x \geq p^{(F,1)}_x$ on these weights while  adjusting
the sample size parameter to $k \sum_x \pi_x$.

\ignore{
  A simple with-replacement sampling scheme samples each key
  independently $k$ times with probability 
$p^{(F,1)}_x$.  Note that the expected total number of 
  samples is $k H(F)$.

Our estimate $\widehat{\Sum}(f;H)$ is then $1/k$ times the sum over sampled keys in $H$ of the number of times a key $x$ was sampled
multiplied by $f_x/p^{(F,1)}_x$.  

 We can also treat the pairs $(x,p^{(F,1)}_x)$ as key value pairs to
 which we apply a weighted sampling scheme.   The sample size that we
 will need will be at most $k'=k H(F)$.
 \begin{lemma} 
The pps probabilities of the pairs $(x,p^{(F,1)}_x)$ with parameter
$k'=k H(F)$ are greater or equal to the multi-objective PPS probabilities of
$(f,k)$.
\end{lemma}
\begin{proof}
$$p^{(F,k)}_x = \min\{1, k
p^{(F,1)}_x\}\ $$ 
\end{proof}

 but an alternative scheme samples
each key independently $k$ times 
according to $p^{(F,1)}_x$, with each positive sample adding
  $g(w_x)/p^{(F,1)}_x$ to the estimate.
}

\subsection{Lower bound on multi-objective sample size} \label{lowerbound:sec}

The following theorem shows that any multi-objective sample for $F$ that meets the quality guarantees on all domain queries for $(f,k_f)$ must include each key $x$ with probability at least $\frac{p^{(f,1)}_x k_f}{p^{(f,1)}_x k_f+1} \geq \frac{1}{2} p^{(f,k_f)}_x$.  Moreover, when $p^{(f,1)}_x k_f \ll 1$, then the lower bound on inclusion probability is close to $p^{(f,k_f)}_x$.
This implies the multi-objective pps sample size is necessary to meet
the quality guarantees we seek (Theorem~\ref{cvbound:thm}).

\begin{theorem} \label{probbound}
  Consider a sampling scheme that for weights $f$ supports estimators that
satisfy for each segment $H$
$$\CV[\widehat{\Sum}(f;H)] \leq
1/\sqrt{q^{(f)}(H)k}\ .$$
Then the inclusion probability of a key $x$ must be at least
$$p_x \geq \frac{p^{(f,1)} k}{p^{(f,1)} k+1} \geq \frac{1}{2} \min\{1, p^{(f,1)} k\}\ .$$
\end{theorem}
\begin{proof}
Consider a segment of a single key $H=\{x\}$.  Then $p^{(f,1)} \equiv q^{(f)}(H) 
=f_x/\sum_y f_y \equiv q$. 
The best nonnegative unbiased sum estimator is the HT estimator: When the key is not sampled, there is no evidence of the segment and the estimate must be $0$.
When it is, uniform estimate minimize the variance.  
If the key is included with probability $p$, the CV of the estimate is 
$$\CV[\widehat{\Sum}(f;\{x\})]=(1/p-1)^{0.5}\ .$$  From the requirement
$(1/p-1)^{0.5} \leq 1/(q k)^{0.5}$, we obtain that 
$p\geq \frac{qk}{qk+1}$. 
\end{proof}

\section{Multi-objective bottom-$k$ samples} \label{bottomkMO:sec}
The  sample $S^{(F)}$
is defined with respect to random $\{u_x\}$.  Each dedicated sample
$S^{(f,k_f)}$ includes the $k_f$ lowest $f$-seeds, computed using
$\{u_x\}$.   $S^{(F)}$ accordingly includes all 
keys that have one of the $k_f$ lowest $f$-seeds for at least one
$f\in F$.

To estimate statistics $\Sum(g;H)$ from bottom-$k$ $S^{(F)}$,  we again 
apply the inverse probability estimator \eqref{MOest} but here we use 
the {\em conditional} inclusion probability $p^{(F)}_x$ for each key $x$
 \cite{multiw:VLDB2009}.  This is 
the probability (over $u_x \sim 
  U[0,1]$) that $x \in S^{(F)}$, when fixing $u_y$ for 
  all $y\not= x$ to be as in the current sample. Note that 
$$p^{(F)}_x = \max_{f\in F} p_{x}^{(f)}\ ,$$ where $p_x^{(f)}$ are as 
  defined in \eqref{inclusiont}.

In order to compute the probabilities $p_x^{(F)}$
for $x\in S^{(F)}$, it always suffices to maintain the slightly larger sample
$\bigcup_{f\in F} S^{(f,k_f+1)}$.  For completeness, we show that it
suffices to instead maintain 
with $S^{(F)} \equiv \bigcup_{f\in F} S^{(f,k_f)}$ a smaller (possibly empty)  
set $Z \subset \bigcup_{f\in F} S^{(f,k_f+1)} \setminus S^{(F)}$ of auxiliary keys.
We now define the set $Z$ and show how inclusion probabilities  can be
computed from $S^{(F)}\cup Z$.
For a key $x\in S^{(F)}$, we denote by
$$g^{(x)} = \arg\max_{f\in F \mid 
  x\in S^{(f)}} p^{(f)}_x$$ the objective with the most forgiving 
threshold for $x$.  If $p^{(g^{(x)})}_x <1$, let $y_x$ be the key with $(k+1)$
smallest $g$-seed (otherwise $y_x$ is not defined).  
The auxiliary keys are then $Z= \{y_x 
\mid x\in S^{(F)}\} \setminus S^{(F)}$.  
We use the sample and auxiliary keys $S^{(F)}\cup Z$  as follows 
to compute the inclusion probabilities:
We first compute for each $f\in F$, $\tau'_f$, which is the
$k_f+1$ smallest $f$-seed of keys in $S^{(F)}\cup Z$.  For each $x\in S^{(F)}$, we then use
$p^{(F)}_x = \max_{f\in F} f(w_x) \tau'_f$ (for priority) or
$p^{(F)}_x = 1- \exp(-\max_{f\in F} f(w_x) \tau'_f)$ (for ppswor).
To see that $p^{(F)}_x$ are correctly computed, note that while
we can have $\tau'_f > \tau^{(f,k_f)}$ for
some $f\in F$ ($Z$ may not include the threshold keys of all the dedicated
samples $S^{(f,k_f)}$),  our definition of $Z$  ensures that $\tau'_f = \tau^{(f,k_f)}$
for $f$ such that
there is at least one $x$ where $f=g^{(x)}$ and $p^{(g^{(x)})}_x < 1$.

\smallskip
\noindent
{\bf Composability:}
Note that multi-objective samples $S^{(F)}$ are composable, since they are a union of (composable) single-objective samples $S^{(f)}$.
It is not hard to see that composability applies with the auxiliary keys:
The set of auxiliary keys in the composed sample must be a subset of sampled and
auxiliary keys in the components.
Therefore, the sample itself includes all the necessary state for
streaming or distributed computation.

\smallskip
\noindent
{\bf Estimate quality:}
 We can verify that for any $f\in F$ and $x$, for any
random assignment $\{u_y\}$ for $y\not=x $,   we have $p^{(F)}_x \geq
p^{(f)}_x$. Therefore (applying Lemma~\ref{domprob:lemma} and noting
zero covariances\cite{bottomk:VLDB2008})  the variance and the CV are at most that of
the estimator \eqref{ppsest} applied to the bottom-$k_f$ sample 
$S^{(f)}$.    To summarize, we obtain the following statistical guarantees on
 estimate quality with multi-objective samples:
\begin{theorem}  \label{MOest:thm}
For each $H$ and $g$, the inverse-probability estimator applied to a
multi-objective pps sample $S^{(F)}$ has
$$\CV[\widehat{\Sum}(g;H)] \leq \min_{f\in F}
\sqrt{\frac{\rho(f,g)}{q^{(g)}(H) k_f}}\ .$$
The estimator applied  to a multi-objective bottom-$k$ samples has the same
guarantee but with $(k_f-1)$ replacing $k_f$.
\end{theorem}

\smallskip
\noindent
{\bf Sample size overhead:}
We must have 
 $\E[|S^{(F)}|] \leq \sum_{f\in F} k_f$. 
The worst-case, where the size of $S^{(F)}$ is 
the sum of the sizes of the dedicated 
 samples, materializes when functions 
 $f\in F$ have disjoint supports.
The sample size,  however, can be much 
 smaller when functions are more related.

 With uniform guarantees ($k_f \equiv k$), we define
 the 
  {\em multi-objective bottom-$k$ overhead} to be 
  $$h^{(\text{botk})}(F) \equiv \E[|S^{(F)}|]/k\ .$$
  This is the sample size overhead of a
  multi-objective versus a dedicated bottom-$k$ sample.

\smallskip
\noindent
{\bf pps versus bottom-$k$ multi-objective sample size:}
For some sets $F$, with the same parameter $k$, we can have much
  larger multi-objective overhead with bottom-$k$ than with pps.
   A multi-objective pps samples is the smallest sample that can
   include a pps sample for each $f$. 
 A multi-objective bottom-$k$ sample must include a bottom-$k_f$ sample for each $f$. 
  Consider a set of $n> k$ keys.  For each subset of $n/2$ keys we define
  a function $f$ that is uniform on the subset and $0$ elsewhere.
  It is easy to see that in this case $h^{(\text{pps})}(F) = 2$ whereas
  $h^{(\text{botk})}(F) \geq (n/2+k)/k$

\smallskip
\noindent
{\bf Computation:}
When the data has the form of
elements $(x,w_x)$ with unique keys and $f_x = f(w_x)$ for a set of functions $F$, then short of further structural assumptions on $F$, the sampling algorithm that computes $S^{(F)}$ must apply 
all functions $f\in F$ to all elements.  The computation is thus
$\Omega(|F|n)$ and can be
$O(|F|n + |S^{(F)}|\log k)$ time by identifying for each $f\in F$,
the $k$ keys with smallest $f$-$\seed{x}$.
In the sequel we will see examples of large or infinite sets $F$ but with special structure that allows us to efficiently compute a multi-objective sample.

\section{The Sampling Closure}  \label{basis:sec}
We define the {\em sampling closure} $\overline{F}$ of a set of
functions $F$ to 
be the set of all functions $f$ such that for all $k$ and for all $H$,
the estimate of  $\Sum(f;H)$ from  $S^{(F,k)}$  has the CV bound 
of Theorem~\ref{cvbound:thm}.
Note that this definitions is with respect to uniform guarantees (same
size parameter $k$ for all objectives).
 We show that the closure of $F$ contains all
non-negative linear combinations of  functions from $F$.
\begin{theorem} \label{closure:thm}
Any $f = \sum_{g\in F}  \alpha_g g$ where $\alpha_g\geq 0$ is in $\overline{F}$.
\end{theorem}
\begin{proof}
We first consider pps samples, where we establish the stronger
claim $S^{(F\cup\{f\},k)}=S^{(F,k)}$, or equivalently, 
\begin{equation} \label{closureclaim}
\text{for all keys $x$,}\  p_x^{(f,k)} \leq p_x^{(F,k)}\ .
\end{equation}
For a function $g$, we use the notation $g({\cal X}) = \sum_{y}
g_y$, and recall that 
 $p_x^{(g,k)} = \min\{1, k \frac{g_x}{g({\cal X})}\}$. 
We first consider $f = c g$ for some $g\in F$.  In this case,
 $p_x^{(f,k)} = p_x^{(g,k)} \leq p_x^{(F,k)}$ and \eqref{closureclaim}
 follows.  
To complete the proof, it suffices to establish \eqref{closureclaim}  for $f=g^{(1)}+g^{(2)}$
such that $g^{(1)}, g^{(2)} \in F$.
Let $c$  be such that 
$\frac{g^{(2)}_x}{g^{(2)}({\cal X})} = c \frac{g^{(1)}_x}{g^{(1)}({\cal X})}$
we can assume WLOG that $c\leq 1$ (otherwise reverse $g^{(1)}$ and $g^{(2)}$).
For convenience denote
$\alpha= g^{(2)}({\cal X})/g^{(1)}({\cal  X})$. 
Then we can write
\begin{eqnarray*}
\frac{f_x}{f({\cal X})} &=& \frac{g^{(1)}_x+
                g^{(2)}_x}{g^{(1)}({\cal X})+g^{(2)}({\cal X})} \\
&=& \frac{(1+c\alpha) g^{(1)}_x}{(1+\alpha) g^{(1)}({\cal X})} =  \frac{1+c\alpha}{1+\alpha} \frac{g^{(1)}_x}{g^{(1)}({\cal X})} \\
&\leq & \frac{g^{(1)}_x}{g^{(1)}({\cal X})}  =
\max\{\frac{g^{(1)}_x}{g^{(1)}({\cal 
      X})},\frac{g^{(2)}_x}{g^{(2)}({\cal X})}   \}\ .
\end{eqnarray*}
Therefore $p_x^{(f,k)} \leq \max\{p_x^{(g^{(1)},k)},p_x^{(g^{(2)},k)}\}\leq p_x^{(F,k)}\ .$

 The proof for multi-objective bottom-$k$ samples is more involved,
 and deferred to the full version.  Note that 
the multi-objective  bottom-$k$ sample $S^{(F,k)}$ may not include
 a bottom-$k$ sample $S^{(f,k)}$, but it is still possible to bound
 the CV.
\end{proof}

\section{The universal sample for monotone statistics} \label{monotone:sec}
In this section we consider the 
(infinite) set $M$ of {\em all} monotone non-decreasing functions and the objectives
$(f,k)$ for all $f\in M$.

We show that the  multi-objective pps and bottom-$k$ samples $S^{(M,k)}$,
which we refer to as the {\em universal monotone sample}, are larger than
a single dedicated weighted sample by at most a logarithmic factor in the number of keys.  We will also show that this is tight.
We will also present efficient universal monotone bottom-$k$ sampling scheme for streamed or distributed data.

  We take the following steps.
We consider the multi-objective sample $S^{(\threshold,k)}$ for the set
 $\threshold$ of all threshold functions (recall that 
$\threshold_T(x) = 1$ if $x \geq T$  and 
  $\threshold_T=0$ otherwise).  We express the inclusion probabilities
  $p^{(\threshold,k)}_x$  and bound the sample size.
  Since all threshold functions are monotone, $\threshold \subset M$.
We will establish that $S^{(\threshold,k)}= S^{(M,k)}$.\footnote{  
  We observe that
  $M = \overline{\threshold}$.  This follows from 
Theorem \ref{closure:thm}, after noticing that
any $f\in M$ can be expressed as a 
non-negative combination of threshold functions $$f(y)= \int_{0}^\infty 
\alpha(T) \threshold_T(y) dT\ ,$$ for some function $\alpha(T)  \geq 0$.  
We establish here the stronger relation $S^{(\threshold,k)}= S^{(M,k)}$.}
We start with the simpler case of pps and then move on to bottom-$k$ samples.

\subsection{Universal monotone pps}

\begin{theorem}
  Consider a data set ${\cal D} = \{(x,w_x)\}$  of $n$ keys and the
sorted order of keys $x$ by non-increasing $w_x$.
  Then a key $x$ that is in position $i$ in the sorted order has
  base multi-objective pps probability
  $$p_x^{(\thresh,1)}= p_x^{(M,1)}\leq 1/i\ .$$
  When all keys have unique weights, equality holds.
\end{theorem}
\begin{proof}
  Consider the function $\thresh_{w_x}$.  The function has weight $1$ on all the $\geq i$ keys with weight $\geq w_x$.  Therefore, the base pps probability is
  $p^{(\thresh_{w_x},1)}_x \leq  1/i$.  When the keys have unique weights then there are exactly $i$ keys $y$ with weight $w_y \geq w_x$ and we have
  $p^{(\thresh_{w_x},1)}_x = 1/i$.
  If we consider all threshold functions, then
$p^{(\thresh_{T},1)}_x = 0$ when $T>w_x$ and $p^{(\thresh_{T},1)}_x \leq 1/i$ when $T\leq w_x$.  Therefore,
  $$p^{(\thresh,1)}_x = \max_T p^{(\thresh_T,1)}_x = p^{(w_x,1)}_x =  1/i\ .$$

  We now consider an arbitrary monotone function $f_x = f(w_x)$.
  From monotonicity there are at least $i$ keys with $f_y \geq f_x$ therefore, $f_x/\sum_y f_y \leq 1/i$.
  Thus, $p^{(f,1)}_x = f_x/\sum_y f_y \leq 1/i$ and
  $$p^{(M,1)}_x = \max_{f\in M} p^{(f,1)}_x \leq 1/i\ .$$
\end{proof}  

There is a simple main-memory sampling scheme where we sort the keys, compute the probabilities $p^{(M,1)}$ and then $p^{(M,k)} = \min\{1, kp^{(M,1)}\}$ and compute a sample accordingly.
We next present universal monotone bottom-$k$ samples and sampling schemes that are efficient on streamed or distributed data.

\subsection{Universal monotone bottom-$k$}

 \begin{theorem}  \label{unimon:thm}
Consider a data set ${\cal D} = \{(x,w_x)\}$ of $n$
keys.   The universal monotone bottom-$k$ sample has expected size 
$\E[|S^{(M,k)}|] \leq k\ln n$ and can be computed using
$O(n + k \log n \log k)$ operations.
 \end{theorem}


 For a particular $T$, the bottom-$k$ sample $S^{(\threshold_T,k)}$
is the set of $k$ keys with smallest
$u_x$ among keys with $w_x \geq T$. 
The set of keys in the multi-objective sample is
$S^{(\threshold,k)}=\bigcup_{T>0} S^{(\threshold_T,k)}$.
We show that a key $x$ is in the multi-objective sample for
$\threshold$ if and only if it is in the bottom-$k$ sample for $
\threshold_{w_x}$:
\begin{lemma} \label{simpleinclusion1:lemma}
Fixing $\{ u_y\}$, for any key $x$,
$$ x\in S^{(\threshold,k)} \iff x\in 
S^{(\threshold_{w_x},k)}\ .$$
\end{lemma}
\begin{proof}
Consider the position $t(x,T)$ of $x$ in an ordering of keys $y$ induced by 
$\threshold_T$-$\seed{y}$.  We claim that  if for a key $y$ we have
$\threshold_T$-$\seed{x} < \threshold_T$-$\seed{y}$ for some $T>0$, this must hold
for $T=w_x$.   The claim can be established by separately considering
$w_y \geq w_x$ and $w_y<w_x$.  The claim implies that $t(x,T)$ is
minimized for $T=w_x$.
\end{proof}

We now consider the
auxiliary keys $Z$ associated with this sample.  Recall that these keys
are not technically part of 
the sample but the information $(u_x,w_x)$ for $x\in Z$ is needed in order 
to compute the conditional inclusion probabilities 
$p^{(\threshold,k)}_x$ for $x\in S$. Note that it follows from
Lemma~\ref{simpleinclusion1:lemma} that for all keys $x$, $p^{(\threshold,k)}_x  =
p^{(\threshold_{w_x},k)}_x $.
For a key $x$, let  $Y_x = \{y\not= x \mid w_y \geq w_x\}$ 
 be the set of keys other than $x$ that have weight that is at least 
 that of $x$.  Let $y_x$ be the key with $k$th smallest $u_x$
 in $Y_x$, when $|Y_x| \geq k$. 
The auxiliary keys are $Z = \{y_x \mid x\in 
S\} \setminus S$. A key $x$ is  included in the sample 
with probability $1$ if $y_x$ is not defined (which means it has one 
of the $k$ largest weights).  Otherwise, it is (conditionally) included 
if and only if $u_x < u_{y_x}$.   To compute the inclusion probability 
$p^{(\threshold,k)}_x$ from $S\cup Z$, we do as follows.  If there are
$k$ or fewer keys in $S\cup Z$ with weight that is at most $w_x$, then
$p^{(\threshold,k)}_x =1$ (For correctness, note that in this case all keys 
with weight $\geq w_x$ would be in $S$.)  Otherwise, observe that
$y_x$ is the key with $(k+1)$th smallest $u$
value in $S\cup Z$ among all keys $y$ with $w_y\geq w_x$. We compute $y_x$ from
the sample and use
$p^{(\threshold,k)}_x = u_{y_x}$.
Note that 
when weights are unique, $Z=\emptyset$.

The definition of $S^{(\threshold,k)}$ is equivalent to that of
an All-Distances Sketch (ADS) computed with respect to 
weights $w_x$ (as inverse distances)
\cite{ECohen6f,ECohenADS:TKDE2015}, and we can apply some 
algorithms and analysis.
In particular, we obtain that 
$\E[|S^{(\threshold,k)}|] \leq  k\ln n$ and the size is
well-concentrated around this expectation.
The argument is simple: Consider keys ordered by decreasing weight.
The probability that the $i$th key has one of the $k$ smallest $u_x$
values, and thus is a member of $S^{(\threshold,k)}$ is at most
$\min\{1,k/i\}$.  Summing the expectations over all keys we obtain
$\sum_{i=1}^n \min\{1,k/i\} <  k\ln n$. 
We shall see that the  bound is
  asymptotically tight when weights are unique. With repeated weights,
  however, the sample size can be much
smaller.\notinproc{\footnote{The sample we would obtain with repeated weights is always a subset of
the sample we would have obtained with  tie 
breaking. In particular, 
the sample size can be at most $k$ times the number 
of unique weights.}}

\begin{lemma}
For any data set $\{(x,w_x)\}$, when using the same randomization
$\{u_x\}$ to generate both samples,
$S^{(M,k)}=S^{(\threshold,k)}$.
\end{lemma}
\begin{proof}
Consider $f\in M$ and the samples obtained for some fixed
randomization $u_y$ for all keys $y$.
Suppose that a key $x$ is in the bottom-$k$ sample $S^{(f,k)}$.  By
definition, we have that $f$-$\seed{x}=r_x/f(w_x)$ is among the $k$ smallest
$f$-seeds of all keys.  Therefore, it must be 
among the $k$ smallest $f$-seeds in the set $Y$ of keys with $w_y \geq w_x$.
From monotonicity of $f$, this implies that $r_x$ must be one of the $k$ smallest in $\{r_y \mid 
y\in Y\}$, which is the same as $u_x$ being one of the $k$ smallest in
$\{u_y \mid y\in Y\}$.   This implies that $x\in S^{(\threshold_{w_x},k)}$.
\end{proof}

\subsection{Estimation quality}

The  estimator \eqref{MOest} with the conditional 
 inclusion probabilities $p^{(M,k)}_x$ 
generalizes the HIP estimator of \cite{ECohenADS:TKDE2015} to
sketches computed for non-unique weights.
Theorem \ref{MOest:thm} implies that for any $f\in M$ and $H$,
$\CV[\widehat{\Sum}(f;H)] \leq \frac{1}{\sqrt{q^{(f)}(H) (k-1)}}\ .$
 When weights are unique and we estimate statistics over all keys,
we have the tighter bound 
$\CV[\widehat{\Sum}(f;{\cal X})] \leq \frac{1}{\sqrt{2k-1}}\ $ \cite{ECohenADS:TKDE2015}.

\subsection{Sampling algorithms}
The samples, including the auxiliary information, are composable.
Composability holds even when
we allow multiple elements to have the same key $x$ and interpret 
$w_x$ to be the maximum weight over elements with key $x$.  
To do so, we use a random hash function to generate $u_x$ consistently for
multiple elements with the same key.
To compose multiple samples, we take a union of the elements, replace
multiple elements with same key with the one of maximum
weight, and apply a sampling algorithm to the set of remaining elements.
The updated inclusion probabilities can be computed from the composed
sample.  

We present two algorithms that compute the sample $S^{(M,k)}$
along with the auxiliary keys $Z$ and the inclusion probabilities
$p^{(M,k)}_x$ for $x\in S^{(M,k)}$.  
The algorithms 
process the elements 
 either in order of decreasing $w_x$ or in order of increasing $u_x$.
\notinproc{
These two orderings may be suitable for different applications and it
is worthwhile to present both: In the related context of all-distances
sketches, both ordering on distances were used \cite{ECohen6f,PGF_ANF:KDD2002,hyperANF:www2011,ECohenADS:TKDE2015}.}
The algorithms are 
correct when applied to any set of $n$ elements that includes 
$S\cup Z$.  
Recall that the inclusion probability $p^{(M,k)}_x$  is the $k$th smallest $u_x$
among keys with $w_x \leq w$.  Therefore, all keys with the 
same weight have the same inclusion probability.  For convenience, we
thus express the probabilities as a function $p(w)$ of the weights.

\begin{algorithm}[h]\caption{Universal monotone
    sampling: Scan by weight \label{univMbyw:alg}}
{\small
Initialize empty max heap $H$ of size $k$ \tcp*{$k$ smallest $u_y$
  values processed so far}
$ptau \gets +\infty$;  \tcp*{** omit with unique weights}
\For{$(x,w_x)$ by decreasing $w_x$ then increasing $u_x$ order}
{
\If{$|H|<k$}{$S \gets S \cup \{x\}$; Insert $u_x$ to $H$; $p(w_x) \gets 1$; Continue}
\eIf{$u_x < \max(H)$}
{$S \gets S \cup \{x\}$; $p(w_x) \gets \max(H)$ \tcp*{$x$ is sampled}
$ptau \gets \max(H)$; $prevw \gets w_x$  \tcp*{**}
Delete $\max(H)$ from $H$\\
Insert $u_x$ to $H$
}
(\tcp*[h]{**}){
\If{$u_x < ptau$ and $w_x = prevw$}
{
$Z \gets Z \cup \{x\}$;
$p(w_x) \gets u_x$
}
}
}
}
\end{algorithm}

Algorithm \ref{univMbyw:alg} processes keys by
order of decreasing weight, breaking ties by increasing $u_x$.
  We maintain a max-heap $H$ of the $k$ smallest $u_y$ values processed so
far. When processing a current key $x$, we include $x\in S$ if
$u_x < \max(H)$.  If including $x$, we delete $\max(H)$ and insert
$u_x$ into $H$.  Correctness follows from $H$ being the $k$ smallest
$u$ values of keys with weight at least $w_x$.  
When weights are unique, 
the probability $p(w_x)$ is the $k$th largest $u$ value in
$H$ just before $x$ is inserted.
When weights are not unique, we also need to compute
$Z$. To do so, we track the previous $\max(H)$, which we call
$ptau$.  If the current key $x$ has $u_x\in (\max(H),ptau)$,  we
include $x\in Z$.  It is easy to verify that in this case, $p(w_x) =
u_x$.  Note that the algorithm may overwrite $p(w)$ multiple times, as
keys with weight $w$ are inserted to the sample or to $Z$.

\notinproc{
Algorithm \ref{univMbyu:alg} processes keys in order of 
increasing $u_x$.  The algorithm maintains a min heap $H$
of the $k$ largest weights processed so far. 
With unique weights, the current key $x$ is included in the sample if
and only if
$w_x > \min(H)$.  If $x$ is included, we delete from the heap $H$ the
key with weight $\min(H)$ and insert $x$.
When weights are not unique, we also track the weight $w$ of the
previous removed key from $H$. When processing a key $x$ then if $w_x
= \min(H)$ and $w_x > prevw$ then $x$ is inserted to $Z$.

The computed inclusion probabilities $p(w)$ with 
unique  weights is $u_y$, were $y$ is the key whose
processing triggered the deletion of the key $x$ with weight $w$ from
$H$.  To establish correctness, consider 
the set $H$, just after $x$ is deleted.
By construction, $H$ contains 
the $k$ keys with weight $w_y<w_x$ that have smallest $u$ values.  
Therefore, $p(w_x) = \max_{y\in H} u_y$.  Since we process keys by increasing
$u_y$, this maximum $u$ value in $H$ is the value of the most recent
inserted key, that is, the key $y$ which
triggered the removal of $x$.
Finally, the keys that remain in $H$ in the end are those with the $k$ largest
weights.  The algorithms correctly assigns $p(w_x)=1$  for these keys.

 When multiple keys can have the same weight, then $p(w)$ is the
 minimum of $\max_{y\in H} u_y$ after the first key of weight $w$ is
 evicted, and the minimum $u_z$ of a key $z$ with $w_z=w$ that was not
 sampled.
If the minimum is realized at such a key $z$, that key is included in
$Z$, and the algorithm set $p(w)$ accordingly when $z$ is processed.
If $p(w)$ is not set already when the first key of weight $w$ is
deleted from $H$, the algorithm correctly assigns $p(w)$ to be 
$\max_{y\in H} u_y$.  After all keys are processed, $p(w_x)$ is set
for remaining keys  $x\in H$ where a key with the same weight was
previously deleted from $H$.  Other keys are assigned $p(w)=1$.

We now analyze the number of operations.
With both algorithms, the cost of processing a 
key is $O(1)$ if the key is not inserted and $O(\log k)$ if the key is
included in the sample.  Using the bound on sample size, we obtain a
bound of $O(n+k \ln n \log k)$ on processing cost.
The sorting requires $O(n\log n)$ computation, which dominates the 
computation (since typically $k\ll n$).  When $u_x$ are assigned
randomly, however, we can generate them 
with a sorted order by $u_x$ in $O(n)$ time, enabling a 
faster $O(n+k\log k \log n)$ computation.

\begin{algorithm}[h]\caption{Universal monotone
    sampling: Scan by $u$ \label{univMbyu:alg}}
{\small
$H \gets \perp$ \tcp*{min heap of size $k$, prioritized by
lex order on $(w_y,-u_y)$, containing keys with largest priorities
processed so far}
$prevw \gets \perp$\\
\For{$x$ by increasing $u_x$ order}
{
\If{$|H|<k$}{$S \gets S \cup \{x\}$;  Insert $x$ to $H$; Continue}
$y \gets \arg\min_{z\in H} (w_z,-u_z)$ \tcp*{The min weight key in $H$ with 
  largest $u_z$}
\eIf{$w_x > w_y$}
{$S \gets S \cup \{x\}$ \tcp*{Add $x$ to sample}
$prevw \gets w_y$\\
\If{$p(prevw) = \perp$}{$p(prevw) \gets u_x$} 
Delete $y$ from $H$\\
Insert $x$ to $H$; 
}
(\tcp*[h]{**}) 
{
\If{$w_x = w_y$ and $w_x > prevw$}
{
$Z \gets Z \cup \{x\}$;  $p(w_x)\gets u_x$
}
}
}
\For(\tcp*[h]{keys with largest weights}){$x\in H$}
{
\If{$p(w_x) = \perp$}{$p(w_x) \gets 1$}
}
}
\end{algorithm}
}

\subsection{Lower bound on sample size}
We now show that the worst-case factor of $\ln n$ on the size of
universal monotone sample is in a sense necessary.  It suffices to
show this for threshold functions:
\begin{theorem} \label{lowerboundMO:lemma}
Consider data sets where all keys have unique weights.
Any sampling scheme with a nonnegative unbiased estimator
that for all $T>0$ and $H$ has $$\CV[\widehat{\Sum}(\threshold_T;H)] \leq
1/\sqrt{q^{(\threshold_T)}(H)k}\ ,$$ must have samples of 
expected size $\Omega(k\ln n)$.  
\end{theorem}
\begin{proof}
We will use Theorem~\ref{probbound} which relates estimation quality to sampling probabilities.
 Consider the key $x$ with the $i$th heaviest weight.
 Applying Theorem~\ref{probbound} to $x$ and $\threshold_{w_x}$
we obtain that $p_x \geq  \frac{k}{k+i}$

 Summing the sampling probabilities over all keys $i\in [n]$  to bound
 the expected sample size we  obtain
$\sum_x p_x \geq k \sum_{i=1}^n \frac{1}{k+i} = k(H_n-H_k) \approx
k(\ln n  - \ln k)\ .$
\end{proof}

\section{The universal capping sample} \label{allcaps:sec}

 An important strict subset of monotone functions is the set
$C = \{\Cap_T \mid T>0 \}$ of capping functions.  We
study the multi-objective bottom-$k$ sample $S^{(C,k)}$, which we refer to as the
{\em universal capping} sample.
\notinproc{  
From Theorem \ref{closure:thm}, the closure of $C$ includes all 
functions of the form 
$f(y)= \int_{0}^\infty 
\alpha(T) \Cap_T(y) dT$, for some $\alpha(T) \geq 0$.  This is 
the set of all non-decreasing concave functions with at most
a linear growth, that is $f(w)$ that satisfy $\frac{d f}{d w} \leq 1$ and 
$\frac{d^2 f}{d w} \leq 0$. 
}

  We show that the sample $S^{(C,k)}$ can be computed using $O(n+k\log n\log k)$
operaions from any
$D' \subset D$ that is 
superset of  the keys in  $S^{(C,k)}$.  \onlyinproc{We defer the
  details to the full version.}\notinproc{
  We start with properties of
  $S^{(C,k)}$ which we will use to design our sampling algorithm.
For a key $x$, let $h_x$ be 
 the number $h$ of keys 
with $w_y\geq w_x$ and
$u_y < u_x$.  Let $\ell_x$ be the number of keys $y$
with $w_y<w_x$ and $r_y/w_y < r_x/w_x$.  

For a key $x$ and $T>0$ and fixing the assignment $\{u_y\}$
for all keys $y$, let $t(x,T)$ be the position of
$\Cap_T$-$\seed{x}$ in the list of values 
$\Cap_T$-$\seed{y}$ for all $y$.   The function $t$ has the following properties:
\begin{lemma} \label{seedorder:lemma}
For a key $x$, $t(x,T)$ is minimized for
$T=w_x$.
Moreover, $t(x,T)$ is non-decreasing for $T \geq w_x$ and
non-increasing for $T\leq w_x$.
\end{lemma}
\begin{proof}
We can verify that 
for any key $y$ 
such that there is a $T>0$ such that the $\Cap_T$-$\seed{x} < \Cap_T$-$\seed{y}$,
we must have 
$\Cap_{w_x}$-$\seed{x} < \Cap_{w_x}$-$\seed{y}$.  
Moreover, the set of $T$ values where $\Cap_{T}$-$\seed{x} <
\Cap_{T}$-$\seed{y}$ is an interval which contains $T=w_x$.
We can establish the claim by
separately considering the cases $w_y \geq w_x$ and $w_y < w_x$.
\end{proof}
As a corollary, we obtain that a key is in the universal capping
sample only if it is in the bottom-$k$ sample for $\Cap_{w_x}$:
\begin{corollary} \label{simpleinclusion:coro}
Fixing $\{ u_y\}$, for any key $x$,
$$ x\in S^{(C,k)} \iff x\in
S^{(\Cap_{w_x},k)}\ .$$
\end{corollary}

\begin{lemma}  \label{lhcond:lemma}
Fixing the assignment $\{u_x\}$, 
$$x \in S^{(C,k)} \iff \ell_x +h_x < k\ .$$
\end{lemma}
\begin{proof}
From Lemma~\ref{seedorder:lemma},
a key $x$ is in a bottom-$k$ $\Cap_T$ sample for some $T$ if and only if it is in the bottom-$k$
sample  for $\Cap_{w_x}$.  
The keys with a lower $\Cap_{w_x}$-seed than $x$ 
are those with $w_y \geq w_x$ and $u_y < u_x$, which are counted in
$h_x$,
and those 
with $w_y < w_x$  and $r_y/w_y < r_x/w_x$, which are counted in
$\ell_x$.
Therefore, a key $x$ is in $S^{(\Cap_{w_x},k)}$ if and only if there are fewer
than $k$ keys with a lower $\Cap_{w_x}$-seed, which is the same as
having $h_x+\ell_x<k$.
\end{proof}
For reference, a key $x$ is in the universal 
monotone sample $S^{(M,k)}$ if and only if it satisfies the weaker 
condition $h_x < k$.

\begin{lemma}
A key $x$ can be auxiliary only if $\ell_x + h_x = k$
\end{lemma}
\begin{proof}
A key $x$ is auxiliary (in the set $Z$) only if for some $y\in S$, it has the $k$th
smallest $\Cap_{w_y}$-seed among all keys other than $y$.  This means
it has the $(k+1)$th smallest $\Cap_{w_y}$-seed.

The number of keys with seed smaller than $\seed{x}$ is
minimized for $T=w_x$.  If the $\Cap_{w_x}$-seed of $x$ is one of the
$k$ smallest ones, it is included in the sample.  Therefore, to be
auxiliary, it must have the $(k+1)$th smallest seed.
\end{proof}

\subsection{Sampling algorithm}
We are ready to  present our sampling algorithm.
We first process the data so that multiple elements with the same key
are replaced with with the one with maximum weight.  
The next step is to identify all keys $x$ with $h_x \leq k$.  
It suffices to compute $S^{(M,k)}$ of the data with the auxiliary
keys.
We can apply a variant of Algorithm \ref{univMbyw:alg}:
We process the keys in  order of decreasing 
weight, breaking ties by increasing rank,
while maintaining a binary search tree $H$ of size $k$ which contains the $k$ lowest $u$ values of 
processed keys.  When processing $x$, if $u_x > \max(H)$ then
$h_x>k$, and the key $x$ is removed. Otherwise, $h_x$
is the position of $u_x$ in $H$, and $u_x$ is inserted to $H$ and 
$\max(H)$ removed from $H$.
  We now only consider keys with $h_x\leq k$.  Note that in
  expectation there are at most $k\ln n$ such keys.

The algorithm then computes 
 $\ell_x$ for all keys with $\ell_x\leq k$.  This is done by scanning keys 
 in order of increasing weight, tracking in a binary search tree structure $H$ the 
(at most) $k$ smallest $r_y/w_y$ values. When processing $x$, if 
$r_x/w_x < \max(H)$, then $\ell_x$ is the  position of $r_x/w_x$ in 
$H$.  We then delete $\max(H)$ and insert $r_x/w_x$.

  Keys that have $\ell_x+h_x<k$ then constitute the sample $S$ and
  keys with $\ell_x+h_x = k$ are retained as potentially being
  auxiliary.

Finally, we perform another pass on the sampled and potentially
auxiliary keys. For each key $x$, we determine the $k+1$th smallest 
$\Cap_{w_x}$-$\seed$, which is $\tau^{(\Cap_{w_x},k)}$.  Using
Corollary \ref{simpleinclusion:coro},  we can use \eqref{inclusiont} to
compute $p_x^{(\Cap_{w_x},k)} = p_x^{(C,k)}$.  At the same time we can
also determine the precise set of auxiliary keys by removing those
that are not the $(k+1)$th smallest seed for any $\Cap_{w_x}$ for
$x\in S$.
}

\notinproc{\subsection{Size of $S^{(C,k)}$}}
  The sample $S^{(C,k)}$ is contained in $S^{(M,k)}$, but can be much 
  smaller.
\notinproc{Intuitively, this is because two keys with similar, but not necessarily identical,  weights are likely to have the 
same relation between their $f$-seeds across all $f\in C$.  This is 
not true for $M$:  For a threshold $T$ between the two weights, the 
$\threshold_T$-seed would always be lower for the higher weight key 
whereas the relation for lower $T$ value can be either one with 
almost equal probabilities.}
\ignore{
spread of weight values is smaller. This is because 
the sample $S^{(M,k)}$ supports 
threshold statistics, and thus must ``separate'' keys 
with very similar weights, as they can have very different 
contributions when $T$ is between their weights.  With 
capping, however, the contributions of keys with similar weights is 
similar across capping parameters. 
}
In particular, we obtain a bound on 
$|S^{(C,k)}|$ which does not depend on $n$:
  \begin{theorem}
$$\E[|S^{(C,k)}|] \leq e k\ln \frac{\max_x w_x}{\min_x w_x}\ .$$
  \end{theorem}
\notinproc{
\begin{proof}
  Consider a set of keys $Y$ such that $\frac{\max_{x\in Y}
    w_x}{\min_{x\in Y}  w_x}  = \rho$.   We show that the expected number of 
keys $x\in Y$ that for at least one $T>0 $ have one of the bottom-$k$
  $\Cap_T$-seeds is at most $\rho k$.  The claim then follows by 
  partitioning keys to $\ln \frac{\max_x w_x}{\min_x w_x}$
groups where weights within each group vary by 
  at most a factor of $e$ and then noticing 
  that the bottom-$k$ across all groups must be a subset of the union 
  of bottom-$k$ sets within each group. 

  We now prove the claim for $Y$.  Denote by $\tau$ the $(k+1)$th 
  smallest $r_x$ value of $x\in Y$.  The set of $k$ keys with $r_x<\tau$ 
  are the bottom-$k$ sample   for cap $T \leq   \min_{x\in Y}w_x$. 
Consider a key $y$.  From Lemma~\ref{seedorder:lemma}, we have 
$y\in S^{(C,k)}$ only if $y\in S^{(\Cap_{w_y},k)}$. 
A necessary condition for the latter is that 
$r_y/w_y < \tau/\min_{x\in Y}w_x$.  This probability  is 
at most 
\begin{eqnarray*}
\lefteqn{\Pr_{u_y \sim U[0,1]} \left[\frac{r_y}{w_y} <
  \frac{\tau}{\min_{x\in Y}w_x} \right]}\\
&\leq & 
\frac{w_y}{\min_{x\in Y}w_x }  \Pr_{u_x\sim U[0,1]} [r_y<\tau]  \leq  \rho \frac{k}{|Y|} 
\end{eqnarray*}
Thus,
the expected number of keys that satisfy this condition is at most 
$\rho k$. 
\end{proof}
}

\section{Metric objectives} \label{metric:sec}
 In this section we discuss the application of multi-objective
 sampling to additive cost objectives.
The formulation has a set of keys $X$, a set of models $\calQ$, and a
nonnegative {\em cost} function $c(Q,x)$ of servicing $x\in X$ by $Q\in \calQ$.
In metric settings, the keys $X$ are points in a metric space $M$,
$Q\in \calQ$ is a configuration of facilities (that
can also be points $Q\subset M$), and $c(Q,x)$ is distance-based and
is the cost of servicing $x$ by $Q$.  For each $Q\in \calQ$ we are
interested in the total cost of servicing $X$ which is
$$c(Q,X) = \sum_{x\in X} c(Q,x)\ .$$
A concrete example is the $k$-means clustering cost function, where
$Q$ is a set of points (centers) of size $k$ and  $c(Q,x) = \min_{q\in Q} d(q,x)^2$.

In this formulation, we are interested in computing a small summary of $X$ that would allow us to estimate $c(Q,X)$ for each $Q\in \calQ$.
Such summaries in a metric setting are sometimes referred to as {\em coresets} \cite{AgarwalCoresets2005}.
Multi-objective samples can be used as such a summary.
Each
$Q\in \calQ$ has a corresponding function $f_x \equiv c(Q,x)$.
A multi-objective sample for the set $F$ of all the functions for $Q\in \calQ$
allows us to estimate $\Sum(f) = c(Q,X)$ for each $Q$.
In particular, a sample of size $h(F)\epsilon^{-2}$ allows us to
estimate $c(Q,X)$ for each $Q\in\calQ$ with CV at most $\epsilon$ and
good concentration.

The challenges, for a domain of such problems, are to
\begin{itemize}
  \item
    Upper-bound the multi-objective overhead $h(F)$ as a  function of
    parameters of the domain ($|X|$, $c$, structure of $Q\in\calQ$).
    The overhead is a fundamental property of the problem domain.
  \item
 Efficiently compute upper bounds $\mu_x$ on the multi-objective
 sampling probabilities $p^{(F,1)}$ so that the sum $\sum_x
 \mu_x$ is not much larger than $h(F)$.  We are interested in
 obtaining these bounds without enumerating over $Q\in \calQ$ (which
 can be infinite or very large).
\end{itemize}


Recently, we \cite{CCK:random15} applied multi-objective sampling
to the problem of {\em centrality} estimation in metric spaces.  Here $M$ is a general
metric space, $X \subset M$ is a set of points,  each $Q$ is a single
point in $M$,  and the cost functions is $c(Q,x) = d(Q,x)^p$.  The
centrality $Q$ is the sum $\sum_x c(Q,x)$.
We established that the multi-objective overhead is constant and that upper bound probabilities
(with constant overhead) can be computed very efficiently using
$O(|X|)$ distance computations.  More recently \cite{CCKcluster17}, we
generalized the result to the $k$-means objective where $Q$ are
subsest of size at most $k$ and
establish that the overhead is $O(k)$.

\section{\eachg, \allg} \label{eachall:sec}
Our multi-objective sampling probabilities
provide statistical guarantees that hold for
{\em each} $f$ and $H$: Theorem~\ref{MOest:thm} states that the estimate $\widehat{\Sum}(f;H)$ has the
CV and concentration bounds over the sample distribution
$S\sim \vecp$ (sample $S$ that includes each $x\in \calX$ 
independently (or VarOpt) with probability $p_x$).

In this section we focus on
uniform  per-objective guarantees ($k_f = k$ for all $f\in F$)  and
statistics $\Sum(f ; \calX) = \Sum(f)$ over the full data set.
For $F$ and probabilities $\vecp$, we define the \eachg\ Normalized Mean Squared Error (NMSE):
  \begin{equation} \label{eachg:eq}
\text{NMSE}_e(F,\vecp)=\max_{f\in F} \E_{S\sim \vecp} (\frac{\widehat{\Sum}(f)}{\Sum(f)}-1)^2  \ ,
  \end{equation}
  and the  \allg\ NMSE:
\begin{equation} \label{allg:eq}
\text{NMSE}_a(F,\vecp)=  \E_{S\sim \vecp} \max_{f\in F}
  (\frac{\widehat{\Sum}(f)}{\Sum(f)}-1)^2\ .   
\end{equation}
The respective normalized root MSE (NRMSE) are the 
squared roots of the NMSE.
Note that \allg\ is stronger than \eachg\ as it 
requires a simultaneous good approximation of $\Sum(f)$ for all 
$f\in F$: 
$$\forall \vecp,\ \text{NMSE}_e(F,\vecp)\leq  \text{NMSE}_a(F,\vecp)\ .$$
\ignore{
  Questions:
  \begin{itemize}
    \item
  Can we show that we always get concentration:
  probability that it exceeds $c\epsilon$ decreases exponentially with $c$.
    \item
      Does the allg samples, scaled by $\epsilon$ have a similar nice formula? of scaling probabilities for $k=1$?
    \item
      We should be able to show the upper bound property also for \allg:  If we increase some probabilities, results only improve.
  \end{itemize}    
}
We are interested in the tradeoff between the
expected size of a sample, which is $\Sum(\vecp) \equiv
\sum_x p_x$, and the NRMSE.
The multi-objective pps probabilities are such that
for all $\ell>1$, 
$\text{NMSE}_e(F, \vecp^{(F,\ell)})\leq 1/\ell$. 

For a parameter $\ell\geq 1$, we can also consider the \allg\ error
$\text{NMSE}_a(F, \vecp^{(F, \ell)})$ 
and ask for a bound on $\ell$ so that the NRMSE$_a \leq \epsilon$.
A union-bound argument established that $\ell=\epsilon^{-2}\log |F|$ always suffices.  Moreover, when $F$ is the sampling closure of a
smaller subset $F'$, then  $\ell =
\epsilon^{-2} \log |F'|$ suffices.
If we only bound the maximum error on any subset of
$F$ of size $m$, we can use $\ell = \epsilon^{-2}\log m$.
When $F$ is the set of 
all monotone functions over $n$ keys, then $\ell = O(\epsilon^{-2}\log\log n)$ suffices.
To see this intuitively, recall that it suffices to consider all threshold functions since all monotone functions are nonnegative combinations of threshold functions.  There are $n$ threshold functions but these functions have $O(\log n)$ points where the value significantly changes by a factor.

  We provide an example of a family $F$ where the sample-size gap
  between $\text{NMSE}_e$ and $\text{NMSE}_a$  is linear in the
  support size.  
Consider a set of $n$ keys and define $f$ for each subset of $n/2$ keys 
so that $f$ is uniform on the subset and $0$ outside it. 
The multi-objective base pps sampling probabilities are $p^{(F,1)}_x =
2/n$ for all $x$ and hence the overhead is $h(F)=2$. Therefore,
$\vecp$ of size $2\epsilon^{-2}$ has $\text{NRMSE}_e(F,\vecp) \leq \epsilon$.
In contrast, any $\vecp$ with $\text{NRMSE}_a(F,\vecp)\leq 1/2$ 
must contain at least one key from the support of each $f$ in (almost)
all samples, implying expected sample size that is at least $n/2$.

 When we seek to bound $\text{NRMSE}_e$, probabilities
 of the form $\vecp^{(F,\ell)}$ essentially optimize the size-quality
 tradeoff.  For $\text{NRMSE}_a$, however, the
minimum size $\vecp$ that meets a certain error
 $\text{NRMSE}_a(F,\vecp) \leq \epsilon$ can be much smaller than the
 minimum size $\vecp$ that is restricted to the form $\vecp^{(F,\ell)}$.
Consider  $\epsilon>0$ and a set $F$ that has $k$ parts $F_i$ with
disjoint supports of equal sizes $n/k$.  All the parts $F_i$ except
$F_1$ have a single $f$ that is uniform on the support,
which means that  with uniform $\vecp$ of size $\epsilon^{-2}$ we have, 
$\text{NRMSE}_e(F_i,\vecp)=\text{NRMSE}_a(F_i,\vecp)=\epsilon$.  The part $F_1$
has similar structure to our previous example which means that any
$\vecp$ that has 
$\text{NRMSE}_a(F_1,\vecp) \leq 1/2$  has size at least $n/(2k)$
wheras a uniform $\vecp$ of size $2\epsilon^{-2}$ has
$\text{NRMSE}_e(F_1,\vecp) = \epsilon$.
 The multi-objective
base sampling probabilities are therefore
$p^{(F,1)} = k/n$
for keys in the supports of $F_i$ where $i>1$ and $p^{(F,1)} = 2k/n$
for keys in the support of $F_1$ and thus the overhead is $k+1$.
The minimum size $\vecp$ for $\text{NRMSE}_a(F,\vecp)=1/2$ must
have value at least $1/2$ for keys in the support of $F_1$ and value
about $(k/n)\log(k)$ for other keys (having \allg\ requirement for
each part and a logarithmic factor due to a (tight) union bound).
Therefore the sample size is $O(n/k + k\log k)$.  In contrast, 
To have
$\text{NRMSE}_a(F,\vecp^{(F,\ell)})=1/2$, we have to use $\ell =
\Omega(n/k)$, obtaining a sample size of $\Omega(n)$.

\section{Optimization over multi-objective samples} \label{opt:sec}

In this section we consider optimization problems where we have keys $\calX$, nonnegative functions $f\in F$ where $f:\calX$ and we seek to maximize
$M(\Sum(f))$  over $f\in F$:
\begin{equation} \label{optobj:eq}
  f = \arg\max_{g\in F} M(\Sum(g))\ .
\end{equation}
We will assume here that the function $M$ is
smooth with bounded rate of change:
$|M(v)-M(v')|/M(v) \leq c |v-v'|/v$, so that when $v' \approx v$ then
$M(v') \approx M(v)$.

Optimization over the large set $\calX$ can be costly or infeasible and we 
 therefore aim to instead compute  an 
 approximate maximizer over a multi-objective sample $S$ of $\calX$.
We propose a framework that adaptively increases the sample size until
the approximate optimization goal is satisfied with the desired
statistical guarantee.

We work with a slightly more general formulation where we allow
the keys to have importance  weights $m_x \geq 0$ and
consider the weighted sums $\Sum(f ; \calX,\vecm) = \sum_{x\in \calX}
f_x m_x$.  Note that for the purpose of defining the
problem, we can without loss of generality ``fold''
the weights $m_x$ into the functions $f\in F$ to obtain the set
$F'\equiv mF$
which has uniform
importance weights so that $f'\in F'$ is defined from $f\in F$ using $\forall x\,  f'_x = f_x m_x$.   When estimating from a
sample, however, keys get re-weighted and therefore it is useful to 
separate out $F$, which may have a particular structure we need to
preserve,  and the importance weights $\vecm$.

Two example problem domains of such optimizations are
clustering (where $\calX$ are points, each $f\in F$ corresponds to a
set of centers,  $f_x$ depend on the distance from $x$ to the
centers,  and $\Sum(f)$ is the cost of clustering with $f$) and
empirical risk minimization
(where $\calX$ are examples and $\Sum(f)$ is the loss of model $f$).
In these settings we seek to minimize ($M(v)=-v$) or maximize
($M(v) = v$) $\Sum(f)$.

 We present a meta-algorithm for approximate optimization
that uses the following:
 \begin{itemize}
   \item
Upper bounds $\pi_x \geq p^{(mF,1)}_x$ on
 the base multi-objective pps probabilities.  We would like $h=\sum_x
 \pi_x$ to be not much larger than $\sum p^{(F,1)}_x$. (Here we denote
 by $mF$ the  importance weights $\vecm$ folded into $F$).
\item
Algorithm $\calA$ that for input  $S\subset \calX$ and
positive weights $a_x$ for $x\in S$ returns $f\in F$ that 
(approximately) optimizes $M(\Sum(f; S,\veca))$.  
By approximate optimum, we allow well-concentrated relative error 
with respect to the optimum $\max_{g\in F}
M(\Sum(g ; S, \veca))$ or with respect to the optimum $\max_{g\in G}
M(\Sum(g ; S, \veca))$   on
a  more restricted set $G\subset F$.

 We apply this algorithm to 
samples $S$  obtained with probabilities $p_x=\min\{1,k\pi_x\}$.  The keys
$x\in S$ have importance  weights $a_x =
m_x/p_x$.  Note that
$\Sum(g ; S,a) = \sum_{x\in S} m_x g_x/p_x$, is the estimate of
$\Sum(g;\calX, \vecm)$ we obtain from $S$.
\ignore{
\item
Efficient algorithm to approximate or compute
$\Sum(f; \calX, \vecm)$ for input $f\in F$.   
}
\end{itemize}

Optimization over the sample requires that 
an (approximate) maximizer $f$ that meets our quality
guarantees over the sample distribution is an approximate
maximizer of \eqref{optobj:eq}.  
Intuitively, we would need 
that at least one approximate maximizer $f$ of \eqref{optobj:eq}
is approximated well by the sample
$M(\Sum(f; S)) \approx M(\Sum(f ; \calX))$ and that
all $f$ that are far from being approximate maximizers are not
approximate maximizers over the sample.

An \allg\ sample is sufficient but
pessimistic.  Moreover, meeting \allg\ typically necessitates worst-case non-adaptive bounds on sample size.
An \eachg\ sample, obtained with $k\geq \epsilon^{-2}$,  is not sufficient in and off itself, but a 
key observation is that maximization over a \eachg\ sample
can only err (within our \eachg\ statistical guarantees) by over-estimating the maximum, that is,
returning $f$  such that
 $M(\Sum(f; S,\veca)) \gg M(\Sum(f ; \calX,\vecm)$.
Therefore, if
\begin{equation} \label{test:eq}
\Sum(f ; \calX,\vecm) \geq (1-\epsilon) \Sum(f; S,\veca)
\end{equation}  
we can certify, within the statistical guarantees provided by $\calA$
and \eachg), that the sample maximizer $f$ is an approximate maximizer
of \eqref{optobj:eq}.
Otherwise, we obtain a lower and approximate upper bounds $$[M(\Sum(f ;
\calX,\vecm)),(1+\epsilon) M(\Sum(f; S,\veca))]$$ on the optimum.
Finally, this certification can be done by exact computation of
$\Sum(f ; \calX,\vecm)$, but it can be performed much more 
efficiently
with statistical guarantees using independent ``validation'' samples from
the same distribution.

  Algorithm~\ref{optsamples:alg} exploits this property to perform
  approximate optimization with an adaptive sample size.
  The algorithm starts with an \eachg\ sample.  It iterates approximate optimization over the sample, testing \eqref{test:eq},
  and doubling the sample size parameter $k$, until the condition \eqref{test:eq}
  holds. Note that since the sample size is doubled, the \eachg\ guarantees
  tighten with iterations, thus, from concentration we get confidence
  for test results over the iterations.
The algorithm uses the smallest sample size where probabilities are of
the form $\min\{1,k\pi_x\}$.   Note (see example in the previous
section) that the optimization might be
supported by a much smaller sample of a different form. An interesting
open question is whether we can devise an algorithm that increases
sampling probabilities in a more targeted way and can perform the
approximate optimization using a smaller
sample size.

\begin{algorithm}[h]\caption{Optimization over multi-objective samples \label{optsamples:alg}}
  \DontPrintSemicolon
  \KwIn{points $\calX$ with weights $\vecm$, $M$, functions $F:\calX$,
    upper bounds $\pi_x \geq p^{(mF,1)}$, algorithm $\calA$ which for
    input $S\subset \calX$ and weights $\veca$ performs
    $\epsilon$-approximate maximization of $M(\Sum(f; S,\veca))$ over $f\in F$.}
 {\small
\ForEach(\tcp*[h]{for sample coordination}){$x\in \calX$}{$u_x \sim  U[0,1]$}
$k \gets \epsilon^{-2}$ \tcp*{Initialize with \eachg\ guarantee} 

\Repeat{$M(\Sum(f;\calX)) \geq (1-\epsilon) M(\Sum(f;S))$ \tcp*{Exact
    or approx using a validation sample }}{
$S \gets \perp$  \tcp*{Initialize empty sample} 
  \ForEach(\tcp*[h]{build sample}){$x \in \calX$ such that $u_x \leq
    \min\{1,k\pi_x\}$}{$S \gets S \cup \{x\}$, $a_x \gets m_x/ \min\{1,k\pi_x\})$}

  \tcp{Optimization over $S$}
Compute $f$ such that $$M(\Sum(f;S,\veca)) \geq (1-\epsilon) \max_{g\in F}  M(\Sum(g;S,\veca))$$\;
$k \gets 2k$ \tcp*{Double the sample size} 
}
\Return{$f$}
}
\end{algorithm}


\section{Conclusion} \label{conclu:sec}
Multi-objectives samples had been studied and applied for nearly five decades.
We present a unified review and extended analysis of multi-objective
sampling schemes, geared towards efficient computation over  very large data sets.  We lay some
foundations for further exploration and additional applications.


A natural extension is the design of efficient multi-objective sampling schemes
for {\em unaggregated} data \cite{FlajoletMartin85,ams99} presented in
a streamed or distributed form.
The data here consists of data elements that are 
key value pairs, where multiple elements can share the same key $x$,
and the weight $w_x$ is the sum of the values of elements with key $x$.  We
are interested again in summaries that support queries of the form
$\Sum(f;H)$, where $f_x = f(w_x)$ for some function $f\in F$.
 To sample unaggregated data, we can first aggregate it and then apply 
 sampling schemes designed for aggregated data. 
Aggregation, however, of streamed or distributed 
data, requires state (memory or communication) of size proportional to the number of unique keys. 
This number can be large, so instead, we aim for efficient sampling without aggregation, using 
state of size proportional to the sample size.  We recently proposed such a sampling framework for 
capping statistics \cite{freqCap:kdd2015}, which also can be used to
for all statistics in their span. 

\ignore{
This is not possible for $\threshold_T$ statistics or frequency 
moments with $p>2$ \cite{ams99}, and therefore, it is not possible to 
compute a universal multi-objective sample in small state over 
unaggregated data.. 
There are, however, dedicated schemes for distinct, sum \cite{}, and 
moments. 
}

{\small 
\bibliographystyle{plain}
\bibliography{cycle} 

\begin{thebibliography}{10}

\bibitem{AgarwalCoresets2005}
P.~K. Agarwal, S.~Har-Peled, and K.~R. Varadarajan.
\newblock Geometric approximation via coresets.
\newblock In {\em Combinatorial and computational geometry, MSRI}. University
  Press, 2005.

\bibitem{ams99}
N.~Alon, Y.~Matias, and M.~Szegedy.
\newblock The space complexity of approximating the frequency moments.
\newblock {\em J. Comput. System Sci.}, 58:137--147, 1999.

\bibitem{hyperANF:www2011}
P.~Boldi, M.~Rosa, and S.~Vigna.
\newblock {HyperANF}: {A}pproximating the neighbourhood function of very large
  graphs on a budget.
\newblock In {\em WWW}, 2011.

\bibitem{BrEaJo:1972}
K.~R.~W. Brewer, L.~J. Early, and S.~F. Joyce.
\newblock Selecting several samples from a single population.
\newblock {\em Australian Journal of Statistics}, 14(3):231--239, 1972.

\bibitem{Cha82}
M.~T. Chao.
\newblock A general purpose unequal probability sampling plan.
\newblock {\em Biometrika}, 69(3):653--656, 1982.

\bibitem{CCK:random15}
S.~Chechik, E.~Cohen, and H.~Kaplan.
\newblock Average distance queries through weighted samples in graphs and
  metric spaces: High scalability with tight statistical guarantees.
\newblock In {\em RANDOM}. ACM, 2015.

\bibitem{ECohen6f}
E.~Cohen.
\newblock Size-estimation framework with applications to transitive closure and
  reachability.
\newblock {\em J. Comput. System Sci.}, 55:441--453, 1997.

\bibitem{ECohenADS:TKDE2015}
E.~Cohen.
\newblock All-distances sketches, revisited: {HIP} estimators for massive
  graphs analysis.
\newblock {\em TKDE}, 2015.

\bibitem{freqCap:kdd2015}
E.~Cohen.
\newblock Stream sampling for frequency cap statistics.
\newblock In {\em KDD}. ACM, 2015.
\newblock full version: {\tt http://arxiv.org/abs/1502.05955}.

\bibitem{CCKcluster17}
E.~Cohen, S.~Chechik, and H.~Kaplan.
\newblock Clustering over multi-objective samples: The one2all sample.
\newblock {\em CoRR}, abs/1706.03607, 2017.

\bibitem{binaryinfluence:CIKM2014}
E.~Cohen, D.~Delling, T.~Pajor, and R.~F. Werneck.
\newblock Sketch-based influence maximization and computation: Scaling up with
  guarantees.
\newblock In {\em CIKM}. ACM, 2014.

\bibitem{varopt_full:CDKLT10}
E.~Cohen, N.~Duffield, C.~Lund, M.~Thorup, and H.~Kaplan.
\newblock Efficient stream sampling for variance-optimal estimation of subset
  sums.
\newblock {\em SIAM J. Comput.}, 40(5), 2011.

\bibitem{topk:conext06}
E.~Cohen, N.~Grossuag, and H.~Kaplan.
\newblock {Processing Top-k Queries from Samples}.
\newblock In {\em Proceedings of the 2006 ACM conference on Emerging network
  experiment and technology (CoNext)}. ACM, 2006.

\bibitem{bottomk07:ds}
E.~Cohen and H.~Kaplan.
\newblock Summarizing data using bottom-k sketches.
\newblock In {\em ACM PODC}, 2007.

\bibitem{bottomk:VLDB2008}
E.~Cohen and H.~Kaplan.
\newblock Tighter estimation using bottom-k sketches.
\newblock In {\em Proceedings of the 34th VLDB Conference}, 2008.

\bibitem{multiw:VLDB2009}
E.~Cohen, H.~Kaplan, and S.~Sen.
\newblock Coordinated weighted sampling for estimating aggregates over multiple
  weight assignments.
\newblock {\em VLDB}, 2(1--2), 2009.
\newblock full: {\tt http://arxiv.org/abs/0906.4560}.

\bibitem{DLT:jacm07}
N.~Duffield, M.~Thorup, and C.~Lund.
\newblock Priority sampling for estimating arbitrary subset sums.
\newblock {\em J. Assoc. Comput. Mach.}, 54(6), 2007.

\bibitem{Feller2}
W.~Feller.
\newblock {\em An introduction to probability theory and its applications},
  volume~2.
\newblock John Wiley \& Sons, New York, 1971.

\bibitem{FlajoletMartin85}
P.~Flajolet and G.~N. Martin.
\newblock Probabilistic counting algorithms for data base applications.
\newblock {\em J. Comput. System Sci.}, 31:182--209, 1985.

\bibitem{GandhiKPS:jacm06}
R.~Gandhi, S.~Khuller, S.~Parthasarathy, and A.~Srinivasan.
\newblock Dependent rounding and its applications to approximation algorithms.
\newblock {\em J. Assoc. Comput. Mach.}, 53(3):324--360, 2006.

\bibitem{HansenH:1943}
M.~H. Hansen and W.~N. Hurwitz.
\newblock On the theory of sampling from finite populations.
\newblock {\em Ann. Math. Statist.}, 14(4), 1943.

\bibitem{HT52}
D.~G. Horvitz and D.~J. Thompson.
\newblock A generalization of sampling without replacement from a finite
  universe.
\newblock {\em Journal of the American Statistical Association},
  47(260):663--685, 1952.

\bibitem{KishScott1971}
L.~Kish and A.~Scott.
\newblock Retaining units after changing strata and probabilities.
\newblock {\em Journal of the American Statistical Association}, 66(335):pp.
  461--470, 1971.

\bibitem{Ohlsson_SPS:1998}
E.~Ohlsson.
\newblock Sequential poisson sampling.
\newblock {\em J. Official Statistics}, 14(2):149--162, 1998.

\bibitem{Ohlsson:2000}
E.~Ohlsson.
\newblock Coordination of pps samples over time.
\newblock In {\em The 2nd International Conference on Establishment Surveys},
  pages 255--264. American Statistical Association, 2000.

\bibitem{PGF_ANF:KDD2002}
C.~R. Palmer, P.~B. Gibbons, and C.~Faloutsos.
\newblock {ANF:} {A} fast and scalable tool for data mining in massive graphs.
\newblock In {\em KDD}, 2002.

\bibitem{Rosen1972:successive}
B.~Ros{\'e}n.
\newblock Asymptotic theory for successive sampling with varying probabilities
  without replacement, {I}.
\newblock {\em The Annals of Mathematical Statistics}, 43(2):373--397, 1972.

\bibitem{Rosen1997a}
B.~Ros{\'e}n.
\newblock Asymptotic theory for order sampling.
\newblock {\em J. Statistical Planning and Inference}, 62(2):135--158, 1997.

\bibitem{Saavedra:1995}
P.~J. Saavedra.
\newblock Fixed sample size pps approximations with a permanent random number.
\newblock In {\em Proc. of the Section on Survey Research Methods}, pages
  697--700, Alexandria, VA, 1995. American Statistical Association.

\bibitem{SSW92}
C-E. S\"{a}rndal, B.~Swensson, and J.~Wretman.
\newblock {\em Model Assisted Survey Sampling}.
\newblock Springer, 1992.

\bibitem{Szegedy:stoc06}
M.~Szegedy.
\newblock The {DLT} priority sampling is essentially optimal.
\newblock In {\em Proc. 38th Annual ACM Symposium on Theory of Computing}. ACM,
  2006.

\bibitem{Tille:book}
Y.~Till\'e.
\newblock {\em Sampling Algorithms}.
\newblock Springer-Verlag, New York, 2006.

\end{thebibliography}
}
\onlyinproc{\end{document}}
\appendix
\section{CV bound by disparity: ppswor} \label{cvdisparity:sec}

\begin{theorem}  \label{ppsworcvrho:thm}
Consider ppswor sampling with respect to weights $f_x$ and the
estimator \eqref{ppsest} computed using \eqref{inclusiont}.
Then for any $g\geq 0$ and segment $H$,
$\CV[\widehat{\Sum}(g;H)] \leq \sqrt{\frac{\rho(f,g)}{q^{(g)}(H) (k-1)}}$.
\end{theorem}
\begin{proof}
  We adapt a proof technique in \cite{freqCap:kdd2015} (which builds on \cite{ECohen6f,ECohenADS:TKDE2015}).
To simplify notation, we use  $W = \Sum(f,{\cal X}) =
\sum_{x\in {\cal X}} f_x$ for the total $f$-weight of the population.   

  We first consider the 
variance of the inverse probability estimate for a key $x$ with 
  weight $g_x$, conditioned on the threshold $\tau$.  
We use the notation $\hat{g}^{(\tau)}_x$ for the estimate that is
$g_x/\Pr[f\text{-}\seed{x} < \tau]$ when $f\text{-}\seed{x}<\tau$ and $0$ otherwise.
Using $p= \Pr[f\text{-}\seed{x}<\tau] = 1-e^{-f_x \tau}$, we have
\begin{eqnarray}
\var[\hat{g}^{(\tau)}_x] &=&  \frac{1-p}{p}g_x^2   = g_x^2
\frac{e^{-\tau f_x}}{1-e^{-\tau f_x}} \nonumber \\
&\leq& \frac{g_x^2}{f_x \tau} \leq \max_{y} \frac{g_y}{f_y} \frac{g_x}{\tau}\ ,\label{varup} 
\end{eqnarray}
using the relation $e^{-z}/(1-e^{-z}) \leq 1/z$. 

 We now consider the variance of the estimator
 $\hat{g}^{(\tau)}_x$ when $\tau$ is the $k$th smallest seed value $\tau'$ in ${\cal 
  X}\setminus x$.  We denote by 
$B_x$ the distribution of $
\tau'$. 
We will bound the variance of the estimate using the relation
$$\var[\hat{g}_x] = \E_{\tau' \sim B_x} \var[\hat{g}^{(\tau')}_x] \ .$$

 The distribution of $\tau'$ is the $k$th smallest  of independent
exponential random  variables with parameters $f_y$ for
 $y\in {\cal X}\setminus x$.
  From properties of the exponential
 distribution, the minimum seed is exponentially distributed with parameter
 $W-f_x$, the difference between the minimum and second smallest is exponentially
 distributed with parameter $W-f_x-w_1$, where $w_1$ is the weight $f_y$ of the
key $y$ with  minimum seed, and so on.   Therefore, the distribution on $\tau'$ conditioned on
 the ordered set of smallest-seed keys is a sum of $k$
 exponential random  variables with parameters {\em at most} $W$. The
 distribution  $B_x$ is a convex combination of such distributions.
We use  the   notation $s_{W,k}$ for the density 
function of the Erlang distribution $\Erlang(W,k)$, which
  is a sum of $k$ independent
  exponential distribution with parameter $W$.
What we obtained is that the distribution $B_x$ (for any $x$) is 
dominated  by $\Erlang(W,k)$.

Since our bound on the conditioned variance $\var[\hat{g}^{(\tau')}_x]$ is non-increasing with
$\tau'$,
domination implies that 
$$\E_{\tau' \sim B_x} \var[\hat{g}^{(\tau')}_x \leq 
\E_{\tau' \sim \Erlang(W,k)} \overline{\var}[\hat{g}^{(\tau')}_x]\ ,$$
where $\overline{\var}$ is our upper bound \eqref{varup}.
 We now use the Erlang density 
function \cite{Feller2} $$s_{W,k}(z) = \frac{W^k z^{k-1}}{(k-1)!}
e^{-Wz}\ $$ 
and the relation
$\int_0^\infty z^a e^{-bz} dz = a!/b^{a+1}$ to bound the variance:
\begin{eqnarray*}
\var[\hat{g}_x]  & \leq  &  \int_0^\infty s_{W,k}(z) \overline{\var}[\hat{g}^{(z)}_x
] d z \\
& \leq &  \int_0^\infty \frac{W^{k} z^{k-1}}{(k-1)!} e^{-Wz}
         \frac{g_x}{z} \max_y \frac{g_y}{f_y} dz \\
& \leq & \max_y \frac{g_y}{f_y}  g_x  \frac{W^{k}}{(k-1)!} \int_0^\infty z^{k-2} e^{-Wz}
         dz \\ & =& \max_y \frac{g_y}{f_y}  \frac{g_x W}{k-1}\ .
\end{eqnarray*}

 By definition, $\widehat{\Sum}(g;H)=\sum_{x\in H} \hat{g}_x$.
 Since  covariances between different keys are zero \cite{bottomk:VLDB2008}, 
$\var[\widehat{\Sum}(g;H)] = \sum_{x\in H} \var[\hat{g}_x] \leq  \max_y \frac{g_y}{f_y} \frac{\Sum(g;H) W}{k-1}$. 
\begin{eqnarray*}
\lefteqn{\CV[\widehat{\Sum}(g;H)]^2 =  \frac{\var[\widehat{\Sum}(g;H)]}{\Sum(g;H)^2}  \leq    \max_y \frac{g_y}{f_y}
                                              \frac{\Sum(g;H) W}{(k-1)
                                              \Sum(g;H)^2}} \\
&\leq& \max_y \frac{g_y}{f_y}  \frac{1}{k-1} \frac{W}{\Sum(g;H)} \\
 & \leq & \max_y \frac{g_y}{f_y}  \frac{1}{k-1} \frac{\Sum(g;\calX)}{\Sum(g;H)} 
\frac{W}{\Sum(g;\calX)} \leq  \frac{\rho}{q(k-1)}\ .
\end{eqnarray*}
\end{proof}

\end{document}